\documentclass[a4paper,UKenglish,cleveref, autoref, thm-restate]{lipics-v2021}

\title{Super-linear Lower Bounds for CSP Non-Redundancy via Shrinking Instances} %

\author{Joshua Brakensiek}{Department of Electrical Engineering and Computer Sciences, University of California, Berkeley, Berkeley, California, USA \and \url{https://jbrakensiek.github.io/} }{josh.brakensiek@berkeley.edu}{https://orcid.org/0000-0003-4149-7298}{Supported in part by the Simons Investigator award of Venkatesan Guruswami and NSF grants CCF-2211972 and DMS-2503280.}

\author{Venkatesan Guruswami}{Simons Institute for the Theory of Computing and the University of California, Berkeley, Berkeley, California, USA \and \url{https://people.eecs.berkeley.edu/~venkatg/} }{venkatg@berkeley.edu}{https://orcid.org/0000-0001-7926-3396}{Supported in part by
a Simons Investigator award and NSF award CCF-2211972.}

\author{Bart M. P. Jansen}{Eindhoven University of Technology, The Netherlands \and \url{https://bjansen.win.tue.nl/}}{b.m.p.jansen@tue.nl}{https://orcid.org/0000-0001-8204-1268}{Supported by the Dutch Research Council (NWO) through Gravitation-grant NETWORKS-024.002.003}

\author{Victor Lagerkvist}{Linköping University, Linköping, Sweden \and \url{https://victorlagerkvist.github.io/}}{victor.lagerkvist@liu.se}{https://orcid.org/0000-0001-8144-4145}{Partially supported by the Swedish research council under grant VR-2025-04487.}

\author{Magnus Wahlström}{Royal Holloway, University of London, London, UK \and \url{https://pure.royalholloway.ac.uk/en/persons/magnus-wahlstr\%C3\%B6m/}}{magnus.wahlstrom@rhul.ac.uk}{https://orcid.org/0000-0002-0933-4504}{}

\authorrunning{J. Brakensiek, V. Guruswami, B. M. P. Jansen, V. Lagerkvist, M. Wahlström} %

\Copyright{Joshua Brakensiek, Venkatesan Guruswami, Bart M. P. Jansen, Victor Lagerkvist, and Magnus Wahlström} %



\category{} %

\acknowledgements{We thank Aaron Putterman for valuable feedback on an earlier draft of this manuscript.}%

\nolinenumbers %

\usepackage{amsmath}
\usepackage{amssymb}
\usepackage{amsthm}
\usepackage{graphicx} %
\usepackage{hyperref}
\usepackage{xcolor}
\usepackage{multicol}
\usepackage{url}
\usepackage[utf8]{inputenc}
\usepackage{enumitem}
\usepackage{cleveref}
\usepackage{doi}
\usepackage{lineno}

\usepackage{tikz}
\usetikzlibrary{matrix,calc,arrows,positioning,graphs}

\date{}

\usepackage{xcolor}
\definecolor{olivegreen}{RGB}{107,142,35} %

\newcommand{\CSP}{\operatorname{CSP}}

\newcommand{\Cat}{\operatorname{Cat}}

\newcommand{\N}{\mathbb{N}}
\newcommand{\F}{\mathbb{F}}

\newcommand{\OR}{\operatorname{OR}}

\newcommand{\eps}{\varepsilon}
\newcommand{\pr}{\operatorname{pr}}

\newcommand{\NRD}{\operatorname{NRD}}

\newcommand{\EQ}{\operatorname{EQ}}

\newcommand{\BCK}{\operatorname{BCK}}

\newcommand{\BoolBCK}{\operatorname{BoolBCK}}
\newcommand{\BoolBCKp}{\operatorname{BoolBCK}^{+}}

\usepackage{xspace}

\begin{document}

\allowdisplaybreaks

\maketitle

\begin{abstract}
We say that an instance of a constraint satisfaction problem (CSP) is non-redundant if the satisfaction of each clause cannot be implied by the satisfaction of the other clauses in the instance. The non-redundancy (NRD) of a CSP is the maximal number of clauses a non-redundant instance can have for a given number of variables. NRD is closely tied to the behavior of CSPs in various computational models including their sparsification, kernelization, and streaming complexity. A primary open question in the study of non-redundancy is the identification of which CSP predicates have near-linear NRD. Recent works by Carbonnel [CP 2022], Khanna, Putterman and Sudan [STOC 2025], Brakensiek and Guruswami [STOC 2025] and Brakensiek, Guruswami, Jansen, Lagerkvist, and Wahlström [2025] have introduced various forms of \emph{gadget reductions} between CSPs to relate their non-redundancy.

The primary contribution of this work is to recontextualize many of these gadget reductions in a framework which we call \emph{hypergraph projections}. By studying a quantity we call the \emph{shrinking factor} of these hypergraph projections, we can more precisely predict when a gadget reduction between predicates can yield a super-linear NRD lower bound, greatly improving on the analysis of previous works. To illustrate the power of our framework, we identify some concrete CSP predicates whose non-redundancy is at the cusp of our understanding and show how our methods give lower bounds that could not have been achieved with previous methods. We also demonstrate how these gadget reductions can be automatically deduced using SAT solvers, thereby opening up novel computational avenues for discovering further relationships between the non-redundancy of various CSPs. %
\end{abstract}

\section{Introduction}

The \emph{constraint satisfaction problem} (CSP) is the well-known problem of deciding whether a set of constraints (often over some fixed template) admits at least one model. CSPs have been studied from many different angles, e.g., classical complexity, approximability, kernelization, sparsification, learnability and query complexity, and more. A recent concept has unified several of these studies: imagine that we want to determine whether a constraint is \emph{non-redundant}, in the sense that it is not implied by any other constraints in the instance. A natural  question is then to determine the maximum number of non-redundant constraints that are possible in an $n$-variable CSP instance. For a fixed relation (sometimes called predicate) $P$ over some finite domain $D$ we let $\NRD(P, n)$ denote this function. As a canonical example, consider the \emph{equality} constraint $\EQ \subseteq D^2$ where $(a,b) \in \EQ$ if and only if $a = b$. Since $\EQ$ has arity $2$, any instance of $\CSP(\EQ)$ can be viewed as a graph $G = (V, E)$, where $V$ are the variables and $E \subseteq \binom{V}{2}$ are the constraints. Observe that if $G$ contains a cycle, then $G$ is redundant---if we satisfy $\EQ$ for all but one edge of the cycle, then transitivity of equality implies that the last edge must also be satisfied. Therefore, maximal non-redundant instances of $\CSP(\EQ)$ are spanning trees; hence, $\NRD(\EQ, n) = n-1$.

NRD was formally introduced by Bessiere, Carbonnel, and Katsirelos~\cite{BCK20} as a tool for understanding a query complexity problem (but its historical roots can be traced back much further~\cite{MONTANARI197495}). However, in the recent years, many different types of applications have been discovered, often hinging on the underlying predicate having sufficiently small NRD.
For example, in cases like $\EQ$, if one can establish that a predicate $P$ has \emph{near-linear} NRD---that is $\NRD(P, n) = n^{1+o(1)}$---then one can build many other (often randomized) structures and algorithms concerning $\CSP(P)$. For example, by a result of Brakensiek and Guruswami~\cite{BG25} one can use non-uniform random sampling to construct a near-linear \emph{sparsifier} of any instance of $\CSP(P)$, analogous to the near-linear cut sparsifiers due to Bencz{\'u}r and Karger~\cite{BK96,K93}. In the setting of kernelization, where one seeks to make an instance of $\CSP(P)$ as small as possible while preserving satisfiability, if the non-redundancy upper bound proof reveals how to constructively detect redundancies, then near-linear kernels can be directly constructed~\cite{C22} (see also discussion in \cite{H26}). More recently, a result due to Sharma and Velusamy~\cite{SV26} shows that near-linear NRD directly implies near-linear single-pass streaming algorithms to decide the satisfiability of $\CSP(P)$. Conversely, if the non-redundancy of $P$ is super-linear, that is $\NRD(P, n) = \Omega(n^{1+\eps})$ for some constant $\eps > 0$, then one can rule out both deterministic and randomized constructions of near-linear sparsifiers~\cite{BG25} and streaming algorithms~\cite{SV26}. Additionally, NRD has been considered in  spatio-temporal reasoning problems for bettering our understanding of branching algorithms~\cite{LGE26} and for the effect that it has has on real-world databases~\cite{DBLP:conf/iske/HuGL00S21,DBLP:conf/ijcai/SioutisLC15}. 

In all these applications the most desirable outcome is to find predicates with linear (or near-linear) NRD. Hence, a central open question in the study of CSP non-redundancy can be stated as follows. 

\begin{center}
\emph{Which CSP predicates have linear or near-linear NRD?}
\end{center}

To highlight the complexity and richness of this question, we first discuss additional examples of predicates with near-linear NRD and then discuss some additional predicates which provably have super-linear NRD.

\paragraph*{Predicates with Linear NRD} 
The canonical example of predicates with linear NRD is that of \emph{linear equations}. 
For example, over the Boolean domain $D = \{0,1\}$, which we can identify with the field $\F_2$, any predicate $P \subseteq \F_2^r$ which is the set of solutions to a system of linear equations over $\F_2$ (e.g., $P$ is an XOR predicate) yields $\NRD(P, n) \le n$. This follows from the fact that in this setting non-redundancy is equivalent to \emph{linear independence} of the constraints, and given there are only $n$ variables, the maximum number of linearly independent constraints is $n$.

A rather subtle consequence of this fact is that, to give a linear upper-bound on NRD, one only needs to show that a predicate $P$ \emph{embeds} into a system of linear equations~\cite{BCK20,CJP20,KPS25,LW20}. To illustrate by example, consider the 1-in-3 SAT predicate $P = \{(0,0,1),(0,1,0),(1,0,0)\}$. To bound the non-redundancy of $P$ consider the linear equations $Q = \{(x,y,z) \in \F_3^3 : x+y+z = 1\}$. By the aforementioned logic, we know that $\NRD(Q, n) \le n$. 
Further observe that $P$ is precisely the set of Boolean tuples in $Q$ (i.e., $P = Q \cap \{0,1\}^3$). 
So if we take a non-redundant instance of $\CSP(P)$ and swap out $P$-constraints with $Q$-constraints, it is easily seen that the resulting $\CSP(Q)$ is also non-redundant, as any logical dependence of $Q$-constraints implies a logical dependence of $P$-constraints.
It follows that $\NRD(P, n) \leq \NRD(Q, n) \leq n$. 
This 1-in-3 SAT example shows that NP-hard predicates can have linear non-redundancy, showing that the study of NRD greatly diverges from the complexity classification of CSPs (e.g., \cite{Bulatov17,zhuk2020Proof}).

More generally, these constructions need not be over a finite field but rather can be done over any group\footnote{However, for non-Abelian examples, one must consider \emph{cosets} rather than \emph{equations}, see~\cite{BG25,KPS25}.} or even more general objects known as \emph{Mal'tsev algebras}. That said, in the Boolean setting, earlier work~\cite{BGJLW25} showed that these complex conditions simplify to the much more elementary concept of being \emph{balanced}: 

\begin{definition}[\cite{CJP20}]
\label{def:balance}
A Boolean predicate $P \subseteq \{0,1\}^r$ is \emph{balanced} if for all~$k \in \mathbb{N}$ and all tuples $t_1, \hdots, t_{2k+1} \in P$ (some possibly equal), if we compute (over $\mathbb Z$) the $r$-tuple
\[
    t := t_1 - t_2 + t_3 - \cdots + t_{2k+1}
\]
and $t \in \{0,1\}^r$, then $t \in P$.
\end{definition}

This condition is equivalent to being definable by a system of equations over a ring $\mathbb Z/m\mathbb Z$~\cite{CJP20,KPS25}. For example, one can easily verify that $P = \{(0,0,1),(0,1,0),(1,0,0)\}$ is balanced. However, the 2-SAT predicate $\OR_2 = \{(0,1),(1,0),(1,1)\}$ is not balanced as
\[
    (0,1) - (1,1) + (1,0) = (0,0) \not\in \OR_2.
\]

\paragraph*{Predicates with Super-linear NRD} As a first example of a predicate with super-linear NRD, we know that $\NRD(\OR_2, n) = \Theta(n^2)$. The reason is that the graph $G = (V, E = \binom{V}{2})$ is a non-redundant instance of $\OR_2$ (or equivalently\footnote{The NRD of a predicate does not change significantly when negations are applied to the variables~\cite{C22}.} 2-SAT). Pick any edge $(u,v) \in E$ and consider an assignment $\psi : V \to \{0,1\}$ such that $\psi(u) = \psi(v) = 0$ but $\psi(w) = 1$ for all other $w \in V \setminus \{u, v\}$. For this assignment, every $\OR_2$ clause of $G$ is satisfied except for $(u,v)$, so $G$ is non-redundant. By similar logic, for the predicate $\OR_k = \{0,1\}^k \setminus \{0^k\}$ (corresponding to $k$-SAT), we have $\NRD(\OR_k, n) = \Omega_k(n^k)$.

With even just this simple family of lower bounds, we can deduce that many other predicates have super-linear NRD via \emph{gadget reductions}~\cite{BGJLW25,C22,KPS25}. These reductions are similar to gadget reductions among computational problems but with an important difference: an ordinary polynomial-time reduction is too blunt since it does not preserve $\NRD$ (consider e.g. the classical reduction from $k$-SAT to $3$-SAT). However, as a trade-off, it is sometimes possible to obtain more powerful classes of reductions since we \emph{only} have to preserve the $\NRD$ metric and do not exactly have to preserve yes- and no-instances.
To illustrate by example, a recent classification by Brakensiek, Guruswami, and Putterman~\cite{BGP26} of Boolean predicates of arity 4, considered as one of their cases the predicate $P = \{(0,0,0,0),(0,0,0,1),(0,1,1,0),(1,1,1,1)\}$. Define a predicate $Q \in \{0,1\}^2$ as a \emph{projection} of $P$ as follows
\begin{align}
    (x_1,x_2) \in Q \iff (0, \bar{x}_1, \bar{x}_1, \bar{x}_2) \in P.\label{eq:A}
\end{align}
With this definition, one can inspect that $Q = \OR_2$, from which one can deduce that $\NRD(P, n) \ge \Omega(\NRD(Q, n)) = \Omega(n^2)$. Also, note that this would not (immediately) give a gadget reduction between the corresponding CSPs since maps of the form $\bar{x}$ are not part of the templates.

A more advanced gadget reduction technique that has recently emerged~\cite{BG25,BGJLW25} allows for multiple variables per argument in the projection. As a key example, consider from \cite{BG25} the predicate $R = \{(x,y,z) \in \F^3_3 \setminus \{(0,0,0)\} : x+y+z = 0\}$. That is, $R$ is a \emph{punctured} set of solutions to a linear equation. A crucial insight by \cite{BG25} is that there exist three functions $g_1, g_2, g_3 : \{0,1\}^2 \to \F_3$ such that
\begin{align}
    \forall x \in \{0,1\}^3, (x_1, x_2, x_3) \in \OR_3 \iff (g_1(x_1,x_2), g_2(x_1,x_3), g_3(x_2,x_3)) \in R.\label{eq:B}
\end{align}
From this, they deduce~\cite{BG25} that $\NRD(R, n) = \Omega(n^{1.5})$. Importantly, they also show that $\NRD(R, n) = O(n^{1.6}\log n)$, so no direct reduction from $\OR_2$ could have worked. See Remark~\ref{remark:3LIN*} for a detailed analysis of this example.

More generally, follow-up work~\cite{BGJLW25} showed that if one can perform gadget reductions as in~(\ref{eq:B}) where $k$ variables are grouped at a time, then one can prove lower bounds of the form $\NRD(P, n) \ge \Omega(\NRD(Q, n)^{1/k})$ from which they could construct for every rational $\beta \in \mathbb Q \cap [1, \infty)$ a predicate $P \subseteq D^r$ (for some sufficiently large values of $|D|$ and $r$) such that $\NRD(P, n) = \Theta(n^{\beta}).$

Another class of methods for lower-bounding NRD introduced by \cite{BGJLW25} borrows from \emph{extremal combinatorics}.\footnote{We note that \emph{upper bounds} using extremal combinatorics have appeared in multiple works prior to \cite{BGJLW25}, see \cite{BCK20,C22}.} Such methods could be viewed as analogues of the extremal combinatorics methods used recently to derandomize special cases of the \emph{range avoidance} problem in complexity theory~\cite{kuntewar2025Range}.
We discuss the relationship between NRD and extremal combinatorics in more detail in our technical overview.

However, despite the power of these gadget reductions, they still do not seem sufficient to fully describe predicates with linear NRD. On the bright side, the aforementioned technique of embedding a predicate into a structure with a sufficiently good algebraic structure (e.g., linear equations, cosets, and so on) is powerful and has been conjectured to fully capture predicates with linear NRD. In the Boolean domain this conjecture is particularly simple since it boils down to whether the predicate is balanced or not.
However, in the absence of such an embedding the existing classes of gadget reductions seem to be too weak to prove super-linear lower bounds. Interestingly, while several different types of such predicates have been identified by different groups of authors~\cite{BCK20,CJP20,LW20,KPS25,BG25,BGJLW25}, they turn out to be equivalent to the predicate\footnote{As is common in the literature, we sometimes omit commas and parentheses when depicting tuples.}
\[\BCK := \{111,222,012,201,120\}\] first identified by Bessiere et al.~\cite{BCK20} (see Brakensiek \& Guruswami for an in-depth discussion~\cite{BG25}).
An important ``Booleanization'' $\BoolBCK \subseteq \{0,1\}^9$ of $\BCK$ can be defined as follows
\begin{align*}
  \BoolBCK := \{010100001, 001010100, 100001010, 001100010, 010001100\}.
\end{align*}

One can think of the tuples in $\BoolBCK$ as the sequences obtained from all $3 \times 3$ permutation matrices \emph{except} the identity matrix, by listing their entries in row-major order. The reason $\BoolBCK$ is imbalanced is an alternating sum of elements of $\BoolBCK$ can produce the identity matrix.
\[
010100001 - 001100010 + 001010100 - 010001100 + 100001010 = 100010001 \not\in \BoolBCK.
\]

Thus, $\BoolBCK$ does not have linear $\NRD$ by any known means, and is particularly noteworthy in the following sense: if a Boolean predicate~$P$ is imbalanced because an alternating sum of \emph{three} tuples of~$P$ gives a 0/1-tuple outside of~$P$, then one can easily show~\cite[Proposition 2.12]{CJP20} that $\OR_2$ is a substructure of $P$ so that~$\NRD(P, n) \geq \NRD(\OR_2, n) \geq \Omega(n^2)$. 
To understand whether imbalanced Boolean predicates have superlinear NRD, the logical next step is to investigate what happens when~$P$ is preserved by all alternating sums of three tuples, but not by a sum of five tuples. This is what $\BoolBCK$ captures, as explained in earlier work~\cite[p.~2240]{CJP20}.\footnote{The predicate $\BoolBCK$ can be obtained by negating the first three coordinates of the predicate~$R^*$ described in~\cite{CJP20}, which does not affect NRD, and then re-ordering columns.}

\subsection{Our Contributions and Technical Overview}

The main technical contribution of this paper is to \emph{revisit} the NRD gadget reduction framework and improve its analysis. With the troublesome predicate $\BoolBCK$ in mind we propose a new gadget framework and show how it can be used to obtain new super-linear lower bounds which were not previously known, and which do not seem to be possible with earlier techniques. Although we do not attack $\BoolBCK$ itself we show that, in a precise technical sense, our gadget framework is strong enough to handle every possible ``restriction'' of $\BoolBCK$. Thus, $\BoolBCK$ is the \emph{simplest} remaining obstacle for describing Boolean predicates with linear NRD. 

\subsubsection{Reductions for NRD via hypergraph projections}
Let us now describe how we obtain this result in greater detail. For a given predicate $P \subseteq D^r$, we can think of a non-redundant instance of $\CSP(P)$ as an $r$-uniform hypergraph\footnote{Technically $H$ is a generalization of a hypergraph,
similar to \emph{directed} hypergraphs, 
in that each tuple has an ordering as implied by $V^r$, but we omit this distinction for simplicity of notation.} $H = (V, E \subseteq V^r)$. By using standard techniques (e.g., \cite{C22}), we can further assume this instance is \emph{$r$-partite} while only changing the non-redundancy by a constant factor. That is, there is a partition $V = V_1 \cup \cdots \cup V_r$ of $V$ such that $E \subseteq V_1 \times V_2 \times \cdots \times V_r$.%

Our goal is to transform $H$ in a way such that it becomes a non-redundant instance of some target predicate $Q \subseteq (D')^{r'}$. Toward this, we consider \emph{projections} of this hypergraph. More precisely, for a set $I \subseteq [r] := \{1, \hdots, r\}$ of indices, we can define an $|I|$-uniform hypergraph $\pi_I H$ whose vertices are $\pi_I V := \bigcup_{i \in I} V_i$ and whose edges are $\pi_I E := \{\pi_I e := (e_i : i \in I) : e \in E\}$. The key idea is that if we ``bundle'' $r'$ of these projections, we can create an instance of $\CSP(Q)$. More precisely, pick $r'$ subsets $I_1, \hdots, I_{r'} \subseteq [r]$ (possibly with repetition) and let $\mathcal I = \{I_1, \hdots, I_{r'}\}$. We then define the \emph{$\mathcal I$-projection} of $H$ to be an $r'$-partite $r'$-uniform hypergraph $H' = (V', E')$ whose vertex set is the disjoint union of $\{\pi_{I_i} E : i \in [r']\}$ and whose edges correspond to $E$, except we replace $e \in E$ with $(\pi_{I_1} e, \hdots, \pi_{I_{r'}} e) \in E'.$

When is this  $H'$ a non-redundant instance of $\CSP(Q)$? It turns out it suffices to check if there is a gadget map $g : D^r \to (D')^{r'}$ with the following properties:
\begin{enumerate}
\item For every $p \in P$, $g(p) \in Q$.
\item For every $p \in D^r \setminus P$, $g(p) \not\in Q$.
\item For every $i \in [r']$, the $i$th coordinate of the map $g$ only depends on the input coordinates indexed by $I_i$.%
\end{enumerate}
If such a map exists, we say that ``$P$ is an $\mathcal I$-substructure of $Q$'' as it shows that a copy of $P$ in a sense ``lives inside'' $Q$. %
Returning to our previous examples, equation (\ref{eq:A}) depicts a $(\{\}, \{1\}, \{1\}, \{2\})$-substructure and equation (\ref{eq:B}) depicts a $(\{1,2\},\{1,3\},\{2,3\})$-substructure.  Whether $P$ is an $\mathcal I$-substructure of $Q$ can be determined within seconds using a SAT solver---see Appendix~\ref{app:sat} for more details. 

If $P$ is an $\mathcal I$-substructure of $Q$, then we can immediately deduce a non-redundancy lower bound on $\CSP(Q)$ from a non-redundant instance $H = (V, E)$ of $\CSP(P)$. %
More precisely, we have that
\[
  \NRD(Q, \sum_{i=1}^{r'}|\pi_{I_i} E|) \ge |E|.
\]
In particular, if we can construct non-redundant instances $H$ of $\CSP(P)$ where $|\pi_{I_i} E| \le |E|^{1-\eps}$ for all $i \in [r']$, then we can deduce that $\NRD(Q, n)$ has super-linear non-redundancy! We define such instances $H$ to be $(|E|^{\eps}, \mathcal I)$-\emph{shrinking instances}. This leads to our main technical result.

\begin{theorem}[Informal version of Theorem~\ref{thm:shrink-lb}]\label{thm:main}
Assume that $P$ is an $\mathcal I$-substructure of $Q$ and that $\CSP(P)$ has an infinite family of $(|E|^{\eps}, \mathcal I)$-\emph{shrinking instances}. Then, $\NRD(Q, n) = \Omega(n^{\frac{1}{1-\eps}})$.
\end{theorem}

\paragraph*{Comparison to Brakensiek et al.~\cite{BGJLW25}} Before we dive into various applications of Theorem~\ref{thm:main}, we must compare our lower bound with the \emph{$c$-fgpp} framework of Brakensiek et al.~\cite{BGJLW25}. Briefly speaking, our projection framework can be viewed as a tighter combinatorial interpretation of their framework. More precisely, in the language used in this paper, the $c$-fgpp reduction framework assumed that each set $I_i$ has the same size $c$. Furthermore, the instances constructed by \cite{BGJLW25} naively have $|V|^c$ vertices rather than the much sharper bound of $\sum_{i=1}^{r'}|\pi_{I_i} E|$. As a result, \cite{BGJLW25} could only prove non-redundancy lower bounds of the form $\NRD(Q, n) \ge \NRD(P, n)^{1/c}$ (see our Corollary~\ref{cor:richness})---meaning the method is only of use in the scenario $\NRD(P, n) = \Omega(n^d)$ for some $d > c$, limiting its utility.
In contrast, we show that in many such situations, super-linear lower bounds are still possible due to the existence of shrinking instances. In other situations where \cite{BGJLW25} can also prove super-linear lower bounds, our construction of shrinking instances can produce even stronger lower bounds (see Remark~\ref{remark:better}).

\subsubsection{Shadows of $\BoolBCK$}

The main application of our reduction framework above is to derive super-linear lower bounds on the NRD of predicates tightly related to $\BoolBCK$ (and which any full classification of linear NRD would eventually need to handle).
These predicates are not arbitrary, but can be systematically obtained as follows.
First, it is known that adding the identity matrix to $\BoolBCK$ makes it balanced. In other words, if we define $\BoolBCKp := \BoolBCK \cup \{100010001\}$ we know that $\NRD(\BoolBCKp, n) = O(n)$. By using the framework of Brakensiek and Guruswami~\cite{BG25}, one can then show that if $\BoolBCK$ has super-linear non-redundancy, then there is an instance with a super-linear number of constraints such that the non-redundancy of each constraint can be witnessed by an assignment which maps that constraint to $100010001$. %
This setup is known as \emph{conditional} NRD and in its general form (See Definition~\ref{def:conditional-nrd} for a formal definition) we have two $r$-ary predicates $S \subseteq T \subseteq D^r$ with the guarantee that any assignment witnessing NRD of $S$ stays inside $T$. The main insight is then that it is much easier to understand NRD with this additional guarantee, and for $\BoolBCK$ we thus conclude that  it suffices to study its conditional NRD with respect to $\BoolBCKp$ (often written $\BoolBCK \mid \BoolBCKp$).

Although we do not know whether a super-linear construction exists, we can still ask about ``shadows'' of its existence. More precisely, consider a $9$-uniform $9$-partite hypergraph $H$ which is a non-redundant instance of $\BoolBCK \mid \BoolBCKp$. Every projection $\pi_I H$ can be viewed as a non-redundant instance of a projected predicate (i.e., $\pi_I \BoolBCK \mid \pi_I \BoolBCKp$). We show that Theorem~\ref{thm:main} is strong enough to deduce that \emph{all} of these projections can have super-linear size.

\begin{theorem}[Informal version of Theorem~\ref{thm:BoolBCK-projections}]\label{thm:boolbck-intro}
Every non-trivial projection of the conditional predicate $\BoolBCK \mid \BoolBCKp$ has non-redundancy $\Omega(n^{1+c})$ for some absolute constant $c> 0$ %
\end{theorem}

To get these lower bounds, we construct a \emph{non-Boolean} arity $4$ conditional predicate $P \mid Q$ corresponding (essentially) to the product of two $6$-cycles. By suitably adapting an extremal combinatorics lower bound of \cite{BGJLW25}, we can show that $\NRD(P \mid Q, n) = \Omega(n^3)$. We then show that all non-trivial projections of $(\BoolBCK | \BoolBCKp)$ have this conditional predicate $P \mid Q$ as an  $\mathcal I$-substructure, where a suitable collection $\mathcal I$ of sets of size $3$ is chosen for each projection. We remark that these substructures were identified via use of a SAT solver (see Appendix~\ref{app:sat} for further details). Due to the parameters of our gadget reduction, the results of \cite{BGJLW25} are insufficient for deducing a non-trivial lower bound. However, we can show that there are infinitely many non-redundant instances of $\CSP(P \mid Q)$ which are $(|E|^{1/6}, I)$-shrinking for every $I \subsetneq [4]$. From this, we can get a super-linear lower bound. %

\paragraph*{Lower Bounds for Non-Boolean Predicates} Similar to Theorem~\ref{thm:boolbck-intro}, we also proved lower bounds on the non-redundancy of various non-Boolean predicates. The particular ones we study in this paper are inspired by the recently-discovered ``Catalan identities'' by \cite{BGJLW25} which can be viewed as a non-Boolean analogue of Definition~\ref{def:balance}. See Appendix~\ref{app:non-Boolean} for these results and further discussion. %

\paragraph*{The Road Ahead} In this paper, we revisited the gadget reduction framework of previous works studying NRD and related problems. Our use of hypergraph projections and $\mathcal I$-substructures shows that the classification of CSP non-redundancy of different predicates is even more interconnected than previously thought. %
Furthermore, the observation that the search for these gadget reductions can be automated using a SAT solver gives further hope that suitable large-scale computations could be applied to locate further relationships between the non-redundancy of various CSPs. In fact, our identification of interesting non-Boolean substructures in projections of $(\BoolBCK | \BoolBCKp)$ shows that a very rich family of predicates are at our disposal for making such deductions. Carrying out such ``computational cartography'' of the landscape of NRD is the subject of future work.

\subsection{Outline}

In Section~\ref{sec:prelims} we formally define all the basic concepts discussed in the introduction, including (conditional) non-redundancy and hypergraph-projections. In Section~\ref{sec:shrink-lb}, we formally define shrinking instances and show how they can be used to construct super-linear non-redundancy lower bounds, establishing Theorem~\ref{thm:shrink-lb}. In Section~\ref{sec:applications}, we apply Theorem~\ref{thm:shrink-lb} to prove Theorem~\ref{thm:main}. In Section~\ref{sec:conclusion}, we give some concluding thoughts and open questions. In Appendix~\ref{app:sat}, we discuss how SAT solvers can be used to construct gadget reductions. In Appendix~\ref{app:conditional}, we discuss how every ``conditional'' non-redundancy problem is asymptotically equivalent (up to a fixed polynomial factor) with a ``traditional'' non-redundancy problem. In Appendix~\ref{app:non-Boolean}, we discuss non-Boolean extensions of Theorem~\ref{thm:boolbck-intro}.

\section{Preliminaries}\label{sec:prelims}

We define a CSP predicate to be a set $P \subseteq D^r$. We call $D$ the \emph{domain} of $P$ and $r$ the \emph{arity} of $P$. We say $P$ is \emph{nontrivial} if $P \neq \emptyset$ and $P \neq D^r$.  We define an $r$-uniform \emph{hypergraph} to be a pair $H = (V, E)$, where $V$ is a set of \emph{vertices} (or \emph{variables}) and $E \subseteq V^r$ is a set of \emph{(hyper)edges} (or \emph{clauses}). We consider the set of all such hypergraphs $H$ to be the \emph{instances} of the constraint satisfaction problem $\CSP(P)$. Note that in this definition, hyperedges are ordered tuples.
Given an $r$-uniform hypergraph $H = (V, E)$ and a domain $D$, we define an \emph{assignment} to be an arbitrary map $\psi : V \to D$. Given a predicate $P \subseteq D^r$, we say a hyperedge $e \in E$ is \emph{$P$-satisfied} (or just \emph{satisfied}) by an assignment $\psi : V \to D$ if $\psi(e) := (\psi(e_1), \hdots, \psi(e_r)) \in P$. We say that $H$ is ($P$-)satisfied by $\psi$ if every $e \in E$ is satisfied by $\psi$.

\begin{definition}\label{def:NRD}
Given a predicate $P \subseteq D^r$, we say that an $r$-uniform hypergraph $H = (V, E)$ is a \emph{non-redundant} (NRD) instance of $\CSP(P)$ if for every edge $e \in E$, there exists an assignment $\psi_e : V \to D$ such that $\psi_e$ satisfies $e'$ for every $e' \in E \setminus \{e\}$, but $\psi_e$ does \emph{not} satisfy $e$. We define $\NRD(P, n)$ to be the size (a.k.a. number of clauses) of the largest non-redundant instance of $\CSP(P)$ on $n$ vertices.
\end{definition}

In logical terms, the condition of NRD is equivalent to none of the clauses of the instance being implied by any of the others. If $P$ is nontrivial, then $\NRD(P, n) \ge \lfloor n/r\rfloor$ by considering a collection of vertex-disjoint hyperedges. As such, linear non-redundancy is the `baseline' for every non-redundancy problem. %

\paragraph*{Conditional non-redundancy} We also make use of the notion of \emph{conditional} non-redundancy coined\footnote{A notion similar to conditional non-redundancy was used as a proof technique in an earlier paper by Bessiere, Carbonnel, and Katsirelos~\cite{BCK20}.} by Brakensiek and Guruswami~\cite{BG25}.

\begin{definition}\label{def:conditional-nrd}
Given predicates $P, Q \subseteq D^r$, we say that $P \mid Q$ is a \emph{conditional predicate} if $P \subseteq Q$. Given $P \mid Q$, we say that an ordered $r$-uniform hypergraph $H = (V, E \subseteq V^r)$ is a \emph{(conditionally) non-redundant} instance\footnote{Think of $\CSP(P \mid Q)$ as a class of optimization problems where for each tuple of $(v_1, \hdots, v_r) \in E$, we impose a \emph{hard} constraint that $(\psi(v_1), \hdots, \psi(v_r)) \in Q$ as well as a \emph{soft} constraint that $(\psi(v_1), \hdots, \psi(v_r)) \in P$.} of $\CSP(P \mid Q)$ if for every edge $e \in E$, there exists an assignment $\psi_e : V \to D$ such that $\psi_e$ $P$-satisfies $e'$ for every $e' \in E \setminus \{e\}$, but $e$ is $Q\setminus P$-satisfied by $\psi_e$. We define $\NRD(P \mid Q, n)$ to be the size of the largest non-redundant instance of $\CSP(P \mid Q)$ on $n$ vertices.
\end{definition}

We note that $\NRD(P \mid D^r, n) = \NRD(P, n)$ as being $D^r \setminus P$-satisfying is equivalent to not being $P$-satisfying. We make use of the following triangle inequality due to \cite[Proposition 7.3]{BG25}.

\begin{lemma}[Triangle inequality~\cite{BG25}]\label{lem:NRD-triangle}
For any predicates $P \subsetneq Q \subsetneq R \subseteq D^r$ and $n \in \N$, we have that
\[
    \NRD(P \mid R, n) \le \NRD(P \mid Q, n) + \NRD(Q \mid R, n).
\]
In particular, if $R = D^r$, we have that
\[
    \NRD(P, n) \le \NRD(P \mid Q, n) + \NRD(Q, n).
\]
\end{lemma}

\begin{remark}
Using techniques similar to \cite[Lemma 4.9]{BGJLW25} and \cite[Proposition 2.4]{BGP26}, one can show that \emph{any} non-trivial conditional predicate $P \mid Q$ has a corresponding ``traditional'' predicate $R$ such that $\NRD(R, n) = \NRD(P \mid Q, n) \cdot \Theta(n^c)$, where $c$ is a fixed constant depending only on $P,Q,R$---we prove this in Appendix~\ref{app:conditional}. Thus, studying conditional non-redundancy is \emph{necessary} toward understanding the broader NRD landscape.
\end{remark}

\paragraph*{$r$-partite instances} We say that an $r$-uniform hypergraph $H = (V, E)$ is \emph{$r$-partite}, if there exists a partition $V_1 \sqcup \cdots \sqcup V_r = V$ of the vertices such that $E \subseteq V_1 \times V_2 \times \cdots \times V_r$. We use the notation $H = (V_1, \hdots, V_r, E)$ to denote an $r$-partite $r$-uniform hypergraph. Given an arbitrary (ordered) hypergraph $H = (V, E \subseteq V^r)$ one can convert it into an $r$-partite hypergraph by replicating the vertex set $r$ times, say $V^{(1)}, \hdots, V^{(r)}$, and replacing each $e \in E$ with an edge $e' \in V^{(1)} \times \cdots V^{(r)}$ where $e_i \in V$ is replaced with the corresponding $e'_i \in V^{(i)}$. By results of \cite{C22,BG25,BGJLW25}, this $r$-partite hypergraph is non-redundant with respect to some predicate whenever the original ordered hypergraph is non-redundant. As such, we can without loss of generality assume that our instance is $r$-partite, affecting our non-redundancy bounds by some multiplicative factor depending on $r$. Thus, we assume for convenience all non-redundant instances we construct in this paper are $r$-partite.

\paragraph*{Projections} Given two sets $I$ and $X$, we define an $I$-tuple over $X$ to be a sequence $(x_i \in X: i \in I)$. We denote the set of all such $I$-tuples by $X^I$. Given a subset $J \subseteq I$, we define $\pi_J : X^I \to X^J$ which ``forgets'' all coordinates not in $J$. This function may be applied to individual $I$-tuples or to a set of $I$-tuples. For example, if we consider $R = \{(0,1,2),(1,2,0),(2,0,0)\}$, then
\[
  \pi_{\{1,3\}} R = \{\pi_{\{1,3\}}(0,1,2),\pi_{\{1,3\}}(1,2,0),\pi_{\{1,3\}}(2,0,0)\} = \{(0,2),(1,0),(2,0)\}.
\]

Given an $r$-partite hypergraph $H = (V_1, \hdots, V_r, E)$ as well as a subset $I \subseteq [r] := \{1,2, \hdots, r\}$, we define $\pi_I H$ to be the $|I|$-partite hypergraph $\pi_I H = ((V_i : i \in I), \pi_I E)$. In other words, we ``forget'' the parts of $H$ not indexed by $I$.

\paragraph*{Gadget Reductions} Previous papers \cite{C22,BG25,BGJLW25} have made use of \emph{functionally guarded primitive promise (fgpp)} definitions (and generalizations) as a type of gadget reduction that can be used to relate the NRD of various CSP predicates. For the lower bounds we prove in this paper, we do not need the full suite of techniques for fgpp, but rather focus on a simpler class of definitions called \emph{substructures} (compare with the less general ``projections'' of \cite{KPS25,BGP26}). %

\begin{definition}\label{def:projection}
Consider predicates $P \subsetneq Q \subseteq D_1^{r_1}$ and $R \subsetneq S \subseteq D_2^{r_2}$. Given a sequence $\mathcal I = (I_1, \hdots, I_{r_2})$ of subsets of $[r_1]$, we say that $P \mid Q$ is an \emph{$\mathcal I$-substructure} of $R \mid S$ if there exist maps $\{g_j : D_1^{I_j} \to D_2 \mid j \in [r_2]\}$ (called the \emph{witnessing maps}) such that  the following two conditions are true.
\begin{align}
  \forall x \in P &, (g_1(\pi_{I_1} x), g_2(\pi_{I_2} x), \hdots, g_{r_2}(\pi_{I_{r_2}} x)) \in R, \label{eq:P}\\
  \forall x \in Q \setminus P &, (g_1(\pi_{I_1} x), g_2(\pi_{I_2} x), \hdots, g_{r_2}(\pi_{I_{r_2}} x)) \in S \setminus R.\label{eq:Q-P}
\end{align}
\end{definition}

In practice, it is convenient to think about an $\mathcal I$-substructure directly as a map $\Sigma : Q \to S$ subject to some restrictions on $\Sigma$ based on the nature of $\mathcal I$. %
We encode this more precisely as follows.
\begin{proposition}\label{prop:projection-compact}
Consider predicates $P \subsetneq Q \subseteq D_1^{r_1}$ and $R \subsetneq S \subseteq D_2^{r_2}$. Given a sequence $\mathcal I = (I_1, \hdots, I_{r_2})$ of subsets of $[r_1]$, we have that $P \mid Q$ is an $\mathcal I$-substructure of $R \mid S$ if and only if there exists a map $\Sigma : Q \to S$ subject to the following conditions
\begin{enumerate}
\item[(1)] $\Sigma(P) \subseteq R$ and $\Sigma(Q \setminus P) \subseteq S \setminus R$.
\item[(2)] For all $i \in [r_1]$ and $j \in [r_2]$, if $i \not\in I_j$ and $x, y \in Q$ differ only in coordinate $i$, then $\Sigma(x)_j = \Sigma(y)_j$.
\end{enumerate}
\end{proposition}

\begin{proof}
First assume that $P \mid Q$ is an $\mathcal I$-substructure of $R \mid S$ with witnessing maps $g_1, \hdots, g_{r_2}$. Define
\[
  \Sigma(x) = (g_1(\pi_{I_1} x), g_2(\pi_{I_2} x), \hdots, g_{r_2}(\pi_{I_{r_2}} x)).
\]
Then (\ref{eq:P}) and (\ref{eq:Q-P}) ensure that (1) holds. Further, since the $j$th coordinate (for $j \in [r_2]$) of $\Sigma(x)$ depends only on $\pi_{I_j} x$, we have that (2) holds as well.

Conversely, given $\Sigma : Q \to S$ subject to (1) and (2), for all $j \in [r_2]$ let $g_j$ be defined such that $g_j(\pi_{I_j} x) = \Sigma_j(x)$. (If some inputs of $g_j$ are undefined, pick them arbitrarily). Condition (2) on $\Sigma$ ensures this choice of $g_j$ is well-defined. Further, condition (1) ensures that (\ref{eq:P}) and (\ref{eq:Q-P}) are satisfied, as desired.
\end{proof}

If $|I_j| = 1$ for each~$j$, then we recover the projections framework of Khanna, Putterman, and Sudan~\cite{KPS25} and a fragment\footnote{The fgpp reduction framework of \cite{C22} and the $c$-fgpp framework of \cite{BGJLW25} also allow for conjugation of predicates. For example, if $(x_1,x_2,x_3) \in R \iff [(x_1,x_2) \in P]  \wedge [(x_2,x_3) \in Q]$ then $\NRD(R, n) \le O(\NRD(P, n) + \NRD(Q, n))$. The hypergraph projections framework in this paper could also be augmented with such operations, but we omit such details for clarity and succinctness of presentation.} of the fgpp-definition framework of Carbonnel~\cite{C22}. If $|I_j| = c$ for some $c \ge 1$, then we recover a fragment of the $c$-fgpp-definition framework of \cite{BGJLW25}. In Section~\ref{sec:shrink-lb}, we show that if $P \mid Q$ is an $\mathcal I$-substructure of $R \mid S$, then a non-redundant instance of $\CSP(P \mid Q)$ can be converted into a non-redundant instance of $\CSP(R \mid S)$.

\begin{remark}\label{remark:3LIN*}
As a non-trivial example of Definition~\ref{def:projection}, consider from \cite{BG25} the example
\begin{align*}
  P &= \{0,1\}^3 \setminus \{(0,0,0)\}\\
  Q &= \{0,1\}^3\\
  R &= \{(0,1,2),(0,2,1),(1,0,2),(1,1,1),(1,2,0),(2,0,1),(2,1,0),(2,2,2)\}\\
  S &= R \cup \{(0,0,0)\}.
\end{align*}
Let $\mathcal I= \{\{1,2\},\{1,3\},\{2,3\}\}.$ We have that $P \mid Q$ is an $\mathcal I$-substructure of $R \mid S$ via the maps $g_1, g_2, g_3$ such that
\begin{align*}
&g_1(0, 0) = 0, g_1(0, 1) = 1, g_1(1, 0) = 2, g_1(1, 1) = 0\\
&g_2(0, 0) = 0, g_2(0, 1) = 1, g_2(1, 0) = 1, g_2(1, 1) = 2\\
&g_3(0, 0) = 0, g_3(0, 1) = 2, g_3(1, 0) = 2, g_3(1, 1) = 1.
\end{align*}
To confirm that (\ref{eq:P}) and (\ref{eq:Q-P}) hold, consult the following table.
\begin{center}
\begin{tabular}{c|c}
$x$ & $(g_1(\pi_{\{1,2\}} x), g_2(\pi_{\{1,3\}} x), g_3(\pi_{\{2,3\}} x))$\\\hline
$(0,0,0)$ & $(0,0,0)$\\
$(0,0,1)$ & $(0,1,2)$\\
$(0,1,0)$ & $(1,0,2)$\\
$(0,1,1)$ & $(1,1,1)$\\
$(1,0,0)$ & $(2,1,0)$\\
$(1,0,1)$ & $(2,2,2)$\\
$(1,1,0)$ & $(0,1,2)$\\
$(1,1,1)$ & $(0,2,1)$
\end{tabular}
\end{center}
The map $x \mapsto (g_1(\pi_{\{1,2\}} x), g_2(\pi_{\{1,3\}} x), g_3(\pi_{\{2,3\}} x))$ is precisely the map $\Sigma$ predicated by Proposition~\ref{prop:projection-compact}. %
In Section~\ref{sec:applications}, we construct much more elaborate reductions.
\end{remark}

\section{NRD Lower Bounds from Shrinking Instances}\label{sec:shrink-lb}

In this section, we prove the main result of this paper, which is a general technique for proving superlinear lower-bounds for CSP non-redundancy. A crucial element of our lower bound technique is that it does not apply to general non-redundant instances, but only to \emph{shrinking instances}.

\begin{definition}\label{def:shrink}
Given an $r$-partite hypergraph $H = (V_1, \hdots, V_r, E)$, a family $\mathcal I$ of subsets of $[r]$, and a parameter $\lambda \ge 1$, we say that $H$ is \emph{$(\lambda, \mathcal I)$-shrinking} if for all $I \in \mathcal I$, we have that $|\pi_I E| \le |E| / \lambda$. We say that $H$ is \emph{$\lambda$-shrinking} if $H$ is $(\lambda, \{I\})$-shrinking for all $I \subsetneq [r]$. 
\end{definition}

In this work we are interested in shrinkage $\lambda = |E|^{\Omega(1)}$. If a non-redundant instance $H = (V, E \subseteq V^r)$ of $\CSP(P \mid Q)$ is $(|E|^{\Omega(1)}, \mathcal I)$-shrinking where $\mathcal I = \{\{i\} : i \in [r]\}$ is the set of all singleton sets, then we can immediately deduce that the non-redundancy of $\CSP(P \mid Q)$ is superlinear. However, shrinkage for larger subsets is useful in reductions, which can be realized through the notion of a \emph{projection hypergraph}. %

\begin{definition}\label{def:projection-hypergraph}
Given an $r$-partite hypergraph $H = (V_1, \hdots, V_r, E)$ and a family of sets $\mathcal I = (I_1, \hdots, I_\ell)$ of subsets of $[r]$ (possibly with repetition), define the $\ell$-partite \emph{projection hypergraph} $\pr_{\mathcal I} H = (\pi_{I_1} E, \hdots, \pi_{I_\ell} E, \pr_{\mathcal I}(E))$, where
\[
  \pr_{\mathcal I}(E) = \{\pr_{\mathcal I}e := (\pi_{I_1} e, \hdots, \pi_{I_\ell} e) \mid e \in E\}.
\]
\end{definition}

We now make the following observation which can be viewed as a subtle but substantial improvement on the techniques of \cite{BGJLW25}.

\begin{lemma}\label{lemma:projection-NRD}
Consider predicates $P \subsetneq Q \subseteq D_1^{r_1}$ and $R \subsetneq S \subseteq D_2^{r_2}$. Given a sequence $\mathcal I = (I_1, \hdots, I_{r_2})$ of subsets of $[r_1]$, assume that $P \mid Q$ is an $\mathcal I$-substructure of $R \mid S$ with witnessing maps $g_1, \hdots, g_{r_2}$. If $H = (V_1, \hdots, V_{r_1}, E)$ is a non-redundant instance of $\CSP(P \mid Q)$, then $\pr_{\mathcal I} H = (\pi_{I_1} E, \hdots, \pi_{I_{r_2}} E, \pr_{\mathcal I}(E))$ is a non-redundant instance of $\CSP(R \mid S)$.
\end{lemma}

\begin{proof}
By the fact that $H$ is a non-redundant instance of $\CSP(P \mid Q)$, for every $e \in E$, there exists $\psi_{e} : V_1 \sqcup \cdots \sqcup V_{r_1} \to D_1$ such that $\psi_e(e') \in P$ for all $e' \in E \setminus \{e\}$ but $\psi_e(e) \in Q \setminus P.$

Let $\hat{V} := \bigsqcup_{i=1}^{r_2} \pi_{I_i} E$ be the vertices of $\pr_{\mathcal I}H$. For each $e \in E$, define $\phi_e : \hat{V} \to D_2$ as follows.
\[
  \forall i \in [r_2], \forall t \in \pi_{I_i}(E),\ \ \ \phi_e(t) := g_i(\psi_e(t)).
\]
In words, we apply $\psi_e$ to the edges of $\pi_{I_i}(E)$, then we convert each assigned edge to a single value with $g_i$. Now, consider an arbitrary $e' \in E$, note that
\begin{align*}
  \phi_e(\pr_{\mathcal I} e') &= (\phi_e(\pi_{I_1} e'), \hdots, \phi_e(\pi_{I_{r_2}} e'))\\
                              &=  (g_1(\psi_e(\pi_{I_1} e')), \hdots, g_{r_2}(\psi_e(\pi_{I_{r_2}} e')))\\
                              &=  (g_1(\pi_{I_1}\psi_e(e')), \hdots, g_{r_2}(\pi_{I_{r_2}} \psi_e(e'))).
\end{align*}
Now, if $e' \in E \setminus \{e\}$, then $\psi_e(e') \in P$. Thus, by (\ref{eq:P}), we have that $\phi_e(\pr_{\mathcal I} e') \in R$. Likewise, if $e' = e$, then $\psi_e(e) \in Q \setminus P$. Thus, by (\ref{eq:Q-P}), we have that $\phi_e(\pr_{\mathcal I} e) \in S \setminus R$. This completes the proof that $\pr_{\mathcal I} H$ is a non-redundant instance of $\CSP(R \mid S)$.
\end{proof}

\begin{remark}
The primary difference between Lemma~\ref{lemma:projection-NRD} and comparable techniques in \cite{BGJLW25} is that \cite{BGJLW25} was inefficient in naively letting $|V|^{|I_i|}$ vertices appear in the $i$th part of $\pr_{\mathcal I}(E)$, whereas Lemma~\ref{lemma:projection-NRD} shows that one can be more careful in the reduction. As shall be seen in Theorem~\ref{thm:shrink-lb} and Section~\ref{sec:applications}, this observation can be the fundamental difference between a trivial lower bound and a superlinear lower bound. %
\end{remark}

We can finally state our main result on deducing superlinear non-redunancy lower bounds.

\begin{theorem}\label{thm:shrink-lb}
Consider predicates $P \subsetneq Q \subseteq D_1^{r_1}$ and $R \subsetneq S \subseteq D_2^{r_2}$. Given a sequence $\mathcal I = (I_1, \hdots, I_{r_2})$ of subsets of $[r_1]$, assume that $P \mid Q$ is an $\mathcal I$-substructure of $R \mid S$. Further let $H_1 = (V_1, E_1), H_2 = (V_2, E_2), \hdots$ be an infinite sequence of non-redundant instances of $\CSP(P \mid Q)$ with the following properties. Assume there exist absolute constants $C_1, C_2, \eps > 0$ (with $\eps < 1$) such that the following conditions hold for all $i \ge 1$:
\begin{enumerate}
  \item $H_i$ is $(C_1 \cdot |E_i|^{\eps}, \mathcal I)$-shrinking, and
  \item $1 \le |E_{i+1}| / |E_{i}| \le C_2$.
\end{enumerate}
Then, $\NRD(R \mid S, n) = \Omega(n^{\frac{1}{1-\eps}}),$ where $\Omega$ hides constant factors dependent on $|E_1|, C_1, C_2, r_1,r_2, \eps.$
\end{theorem}

\begin{proof}
By Lemma~\ref{lemma:projection-NRD}, we have that $\pr_{\mathcal I} H_i$ is a non-redundant instance of $\CSP(R \mid S)$ for all $i \ge 1$. For all $i \ge 1$, define
\begin{align*}
  m_i &= |E_i|,\\
  n_i &= \sum_{j=1}^{r_2} |\pi_{I_j} E_i| \le \frac{r_2}{C_1} |E_i|^{1-\eps} = \frac{r_2}{C_1} m_i^{1-\eps}.
\end{align*}
In particular, $\NRD(R \mid S, \frac{r_2}{C_1} m_i^{1-\eps}) \ge \NRD(R \mid S, n_i) \ge m_i.$ By assumption, we have that $1 \le m_{i+1} / m_i \le C_2$ for all $i \ge 1.$ Therefore, for any $n \ge \frac{r_2}{C_1} m_1^{1-\eps} = O(1)$, we have that there exists $i \ge 1$ such that
\[
  n \in \left[\frac{r_2}{C_1} m_i^{1-\eps}, \frac{r_2}{C_1}C_2 m_{i}^{1-\eps}\right]
\]
Hence,
\[
  \NRD(R \mid S, n) \ge \NRD\left(R \mid S, \frac{r_2}{C_1} m_i^{1-\eps}\right) \ge m_i \ge \Omega(n^{\frac{1}{1-\eps}}),
\]
as desired.
\end{proof}

As an immediate corollary, we can recover the primary content of \cite[Theorem 7.14]{BGJLW25}.

\begin{corollary}\label{cor:richness}
Consider predicates $P \subsetneq Q \subseteq D_1^{r_1}$ and $R \subsetneq S \subseteq D_2^{r_2}$ with $P,R \neq \emptyset$. Given a sequence $\mathcal I = (I_1, \hdots, I_{r_2})$ of subsets of $[r_1]$ of size at most $c$, assume that $P \mid Q$ is an $\mathcal I$-substructure of $R \mid S$. Assume that $\NRD(P \mid Q, n) = \Omega(n^{\alpha})$, then
\[
  \NRD(R \mid S, n) = \Omega(n^{\alpha/c}).
\]
\end{corollary}

\begin{proof}
Since $R$ is non-trivial, a lower bound of $\NRD(R \mid S, n) = \Omega(n)$ is immediate by considering an instance with $n / r_2$ disjoint hyperedges. Thus, we may assume that $\alpha > c$.

For a fixed $n \in \mathbb N$, let $H_n = (V_n, E_n)$ be a non-redundant instance of $\CSP(P \mid Q)$ with $|V_n| = n$ vertices and $|E_n| = \Theta(n^{\alpha})$ edges (and also assume the number of edges is non-decreasing as $n$ increases). For any $I \subseteq [r_1]$ of size at most $c$, we have  $|\pi_I E_n| \le n^{c}$. Thus,
\[
  \frac{|\pi_i E_n|}{|E_n|} = O\left(\frac{n^c}{n^{\alpha}}\right) = O(|E_n|^{c/\alpha - 1}).
\]
That is, $H_n$ is $(\Omega(|E_n|^{1 - c/\alpha}), \mathcal I)$-shrinking. Furthermore, we can see that $\frac{|E_{n+1}|}{|E_n|} = \Theta(\frac{(n+1)^{\alpha}}{n^{\alpha}}) = \Theta(1)$. Therefore, the hypothesis of Theorem~\ref{thm:shrink-lb} for $\eps = 1 - c/\alpha$ hold. Thus,
\[
  \NRD(R \mid S, n) = \Omega(n^{\frac{1}{1-\eps}}) = \Omega(n^{\alpha/c}),
\]
as desired.
\end{proof}

\section{Applications}\label{sec:applications}

In this section, we repeatedly apply Theorem~\ref{thm:shrink-lb} to prove new lower bounds for a variety of (conditional) CSP predicates.

\subsection{Shrinking Instances via Extremal Combinatorics}

We now construct $\lambda$-shrinking non-redundant instances for two conditional predicates.%
Adopting notation from \cite{BGJLW25}, define $C_6$ to be the $6$-cycle written as a binary predicate of domain size $3$:
\begin{align*}
  C_6 &:= \{(0,0),(0,1),(1,0),(1,2),(2,1),(2,2)\}\\
  C_6^{*} &:= C_6 \setminus \{(0,0)\}.
\end{align*}

We recall the following fact about $\CSP(C_6^{*} \mid C_6)$.

\begin{proposition}[Lemma 4.3~\cite{BGJLW25}]\label{prop:girth-6}
A bipartite graph $G = (V_1, V_2, E)$ is a non-redundant instance of $\CSP(C_6^{*} \mid C_6)$ if and only if $G$ has girth at least $6$.
\end{proposition}

It is well-known that there are bipartite graphs of girth $6$ on $n$ vertices and $\Omega(n^{1.5})$ edges~\cite{erdos1964problem}, so $\NRD(C_6^* \mid C_6, n) \ge \Omega(n^{1.5})$. We use this fact in the construction of more elaborate predicates.

\begin{definition}\label{def:box-product}
Given predicates $P_1 \subsetneq Q_1 \subseteq D^{r_1}$ and $P_2 \subsetneq Q_2 \subseteq D^{r_2}$, we define their \emph{box product} $(P_1 \mid Q_1) \boxtimes (P_2 \mid Q_2) = (P_3 \mid Q_3)$ with $P_3 \subsetneq Q_3 \subseteq D^{r_1+r_2}$ as follows.
\begin{align*} 
  Q_3 &:= Q_1 \times Q_2 = \{(t_1, \hdots, t_{r_1}, t'_{1}, \hdots, t'_{r_2}) : t \in Q_1, t' \in Q_2\}\\
  P_3 &:= (P_1 \times Q_2) \cup (Q_1 \times P_2). 
\end{align*}
\end{definition}

Using these relations, we build two conditional predicates as follows.
\begin{align*}
 R_1 \mid S_1 &:= (C_6^{*} \mid C_6) \boxtimes (\{1,2\} \mid \{0,1,2\})\\
 R_2 \mid S_2 &:= (C_6^{*} \mid C_6) \boxtimes (C_6^{*} \mid C_6)
\end{align*}
\begin{remark}
The conditional predicate $R_1 \mid S_1$ is only needed for lower bounds of non-Boolean instances considered in Appendix~\ref{app:non-Boolean}. However, the construction of non-redundant shrinking instances of $\CSP(R_1 \mid S_1)$ is a good warm-up to the construction of non-redundant shrinking instances of $\CSP(R_2 \mid S_2)$, so we present both in this section. %
\end{remark}

In the next two lemmas, we construct shrinking instances of each.
\begin{lemma}\label{lem:R1S1}
For all but finitely many $m$, $\CSP(R_1 \mid S_1)$ has non-redundant instances with $m$ edges which are $O(m^{1/4})$-shrinking.
\end{lemma}
\begin{proof}
Pick an integer $n \ge 1$ and pick sets $V_1, V_2, V_3$ of size $n^2, n^2$, and $n$, respectively. Pick $E_{12} \subseteq V_1 \times V_2$ of size $\Theta(n^3)$ such that $(V_1, V_2, E_{12})$ has girth $6$. Then, let $E = E_{12} \times V_3$. By adapting the proof of \cite[Lemma 4.9]{BGJLW25}, we can show that $H_n := (V_1, V_2, V_3, E)$ is a non-redundant instance of $\CSP(R_1 \mid S_1)$. More precisely, given $e \in E$, we construct $\psi_e : V_1 \sqcup V_2 \sqcup V_3 \to \{0,1,2\}$ as follows. For $v \in V_1 \sqcup V_2$, assign $\psi_e(v)$ the assignment which $v$ is given for the edge $(e_1, e_2)$ as promised by Proposition~\ref{prop:girth-6} for the instance $(V_1, V_2, E_{12})$ of $\CSP(C_6^* \mid C_6)$. Finally, for $v \in V_3$, let $\psi_e(v)$ be $0$ if $v = e_3$ and $1$ otherwise. These assignments show that $H_n$ is a non-redundant instance of $\CSP(R_1 \mid S_1)$.

Now, observe that 
\begin{align*}
|E| &= |E_{12}| \cdot |V_3| = \Theta(n^4)\\
|\pi_{\{1,2\}} H_n| &= |E_{12}| = \Theta(n^3)\\
|\pi_{\{1,3\}} H_n| = |\pi_{\{2,3\}} H_n| &\le n^2 \cdot n = O(n^3).
\end{align*}

Thus, $H_n$ is $O(n^{-1}) = O(|E|^{-1/4})$-shrinking, as desired. To get the hypothesis of ``all but finitely many $m$'', we pick a least $n$ such that $H_n$ has at least $m$ edges and then delete $o(m)$ edges to get an instance with exactly $m$ edges. Since we are only deleting edges, the shrinking factor will change by at most a constant factor.
\end{proof}

\begin{lemma}\label{lem:R2S2}
For infinitely many $m$, $\CSP(R_2 \mid S_2)$ has non-redundant instances with $m$ edges which are $O(m^{1/6})$-shrinking.
\end{lemma}

\begin{proof}
Pick an integer $n \ge 1$ and pick sets $V_1, V_2, V_3, V_4$ each of size $n$. Pick $E_{12} \subseteq V_1 \times V_2$ of size $\Theta(n^{1.5})$ such that $(V_1, V_2, E_{12})$ has girth $6$. Likewise, pick $E_{34} \subseteq V_3 \times V_4$ of size $\Theta(n^{1.5})$ such that $(V_3, V_4, E_{34})$ has girth $6$.

Then, let $E = E_{12} \times E_{34}$. We claim that $H_n := (V_1, V_2, V_3, V_4, E)$ is a non-redundant instance of $\CSP(R_2 \mid S_2)$. For any $e \in E$, we construct $\psi_e : V_1 \sqcup V_2 \sqcup V_3 \sqcup V_4 \to \{0,1,2\}$ as follows. For $v \in V_1 \sqcup V_2$, assign $\psi_e(v)$ the assignment which $v$ is given for the edge $(e_1, e_2)$ as promised by Proposition~\ref{prop:girth-6} for the instance $(V_1, V_2, E_{12})$ of $\CSP(C_6^* \mid  C_6)$. Likewise, for $v \in V_3 \sqcup V_4$, assign $\psi_e(v)$ the assignment which $v$ is given for the edge $(e_3, e_4)$ as promised by Proposition~\ref{prop:girth-6} for the instance $(V_1, V_2, E_{34})$ of $\CSP(C_6^* \mid C_6)$. To see why $\psi_e$ has the desired property, note that $e' \in E$ is not satisfied by $R_2$ if and only if both $\psi_e(e'_1, e'_2) = (0,0)$ and $\psi_e(e'_3, e'_4) = (0,0)$, which happens only when $e = e'$. We have thus shown that $H_n$ is a non-redundant instance of $\CSP(R_1 \mid S_1)$.

Now, observe that $|E| = |E_{12}| \cdot |E_{34}| = \Theta(n^3)$. Furthermore, for any $I \subseteq [4]$ of size $3$, we have that $|\pi_I E| = n \cdot \Theta(n^{1.5}) = \Theta(n^{2.5})$. Thus, $H_n$ is $O(n^{-1/2}) = O(|E|^{-1/6})$-shrinking, as desired.
\end{proof}

\subsection{Boolean Lower Bounds}

Adapting notation from~\cite{BG25}, we let $\BoolBCKp \subset \{0,1\}^9$ denote the set of all $3 \times 3$ permutation matrices. More precisely, $\BoolBCKp = \{t^e, t^{(12)}, t^{(13)}, t^{(23)}, t^{(123)}, t^{(321)}\}$ where
\begin{align*}
t^{(e)} &= 100010001\\
t^{(12)} &= 010100001\\
t^{(13)} &= 001010100\\
t^{(23)} &= 100001010\\
t^{(123)} &= 001100010\\
t^{(321)} &= 010001100
\end{align*}
We then let $\BoolBCK := \{t^{(12)}, t^{(13)}, t^{(23)}, t^{(123)}, t^{(321)}\}$. As verified in \cite{CJP20,KPS25,BG25}, we have that $\BoolBCKp$ is a \emph{balanced} predicate. That is, for every $n$ and every sequence $t^1, t^2, \ldots, t^{2n}, t^{2n +1} \in \BoolBCKp$ the alternating sum (applied component-wise) $t^1 - t^2 + t^3 \ldots - t^{2n} + t^{2n+1} \in \BoolBCKp$, and thus $\NRD(\BoolBCKp, n) = O(n)$~\cite{CJP20}. Therefore, by the triangle inequality, we have that
\begin{align*}
    \NRD(\BoolBCK, n) &\le \NRD(\BoolBCK \mid \BoolBCKp, n) + \NRD(\BoolBCKp, n)\\
        &\le \NRD(\BoolBCK \mid \BoolBCKp, n) + O(n).
\end{align*}
It is straightforward to verify that $ \NRD(\BoolBCK, n) \ge \NRD(\BoolBCK \mid \BoolBCKp, n) = \Omega(n)$ (e.g., consider $n/3$ disjoint constraints), so we have that
\[
  \NRD(\BoolBCK, n) = \Theta(\NRD(\BoolBCK \mid \BoolBCKp, n)).  
\]
Numerous papers have asked about the non-redundancy of $\BoolBCK$ \cite{CJP20,LW20,KPS25,BG25,BGJLW25} or closely related predicates. In particular, due to the lack of being balanced, it is very plausible that $\BoolBCK$ (and thus $\BoolBCK \mid \BoolBCKp$) has a truly superlinear non-redundancy lower bound. We do not know how to prove such a bound using Theorem~\ref{thm:shrink-lb}. However, we can use Theorem~\ref{thm:shrink-lb} to prove a superlinear lower bound for every \emph{projection} of $\BoolBCK \mid \BoolBCKp$.

\begin{theorem}\label{thm:BoolBCK-projections}
For every $J \subsetneq [9]$ such that $\pi_J\BoolBCKp \neq \pi_J\BoolBCK$, we have that
\begin{align}
\NRD(\pi_J\BoolBCK \mid \pi_J\BoolBCKp, n) \ge \Omega(n^{6/5}).\label{eq:BoolBCK-proj-lb}
\end{align}
\end{theorem}

\begin{proof}
For any $I \subseteq J \subseteq [9]$, assuming $\pi_I\BoolBCKp \neq \pi_I\BoolBCK$, we can see that $\NRD(\pi_I\BoolBCK \mid \pi_I\BoolBCKp, n) \ge \Omega(\NRD(\pi_J\BoolBCK \mid \pi_J\BoolBCKp, n))$ by a simple gadget reduction: if $H$ is an $r$-partite non-redundant instance of $\CSP(\pi_J\BoolBCK \mid \pi_J\BoolBCKp)$, then $\pi_I H$ is a non-redundant instance of $\CSP(\pi_I\BoolBCK \mid \pi_I\BoolBCKp)$.

Thus, it suffices to prove (\ref{eq:BoolBCK-proj-lb}) in the case that $|J| = 8$. By applying a suitable (common) permutation to the coordinates of $\BoolBCK$ and $\BoolBCKp$, it suffices to actually prove (\ref{eq:BoolBCK-proj-lb}) in two cases: $J_1 := [9] \setminus \{1\}$ and $J_2 := [9] \setminus \{2\}$, see Lemma~\ref{lem:J-sym}. %

By Lemma~\ref{lem:J1} and Lemma~\ref{lem:J2}, we have that $(C_6^* \mid C_6) \boxtimes (C_6^* \mid C_6)$ is a suitable substructure of both $(\pi_{J_1}\BoolBCKp \mid \pi_{J_1}\BoolBCK)$ and $(\pi_{J_2}\BoolBCKp \mid \pi_{J_2}\BoolBCK)$.  By applying Lemma~\ref{lem:R2S2} to Theorem~\ref{thm:shrink-lb} with the parameter choice of $\eps = 1/6$, we thus get our $\Omega(n^{6/5})$ lower bound for both predicates, as desired.
\end{proof}

We now turn to proving the symmetry-breaking Lemma~\ref{lem:J-sym}.

\begin{lemma}\label{lem:J-sym}
For every $J \subseteq [9]$ of size $8$, either
\begin{align*}
\NRD(\pi_J\BoolBCK \mid \pi_J\BoolBCKp, n) &= \NRD(\pi_{J_1}\BoolBCK \mid \pi_{J_1}\BoolBCKp, n), \text{or}\\
\NRD(\pi_J\BoolBCK \mid \pi_J\BoolBCKp, n) &= \NRD(\pi_{J_2}\BoolBCK \mid \pi_{J_2}\BoolBCKp, n),
\end{align*}
where $J_1 = [9] \setminus \{1\}$ and $J_2 = [9] \setminus \{2\}$.
\end{lemma}

\begin{proof}
Given a conditional predicate $P \subsetneq Q \subseteq D^r$ and a permutation $\sigma : [r] \to [r]$, we can define a permuted instance $P^{\sigma} \subsetneq Q^{\sigma} \subseteq D^r$ by
\begin{align*}
  P^{\sigma} &= \{(x_{\sigma(1)}, \hdots, x_{\sigma(r)}) : x \in P\},\\
  Q^{\sigma} &= \{(x_{\sigma(1)}, \hdots, x_{\sigma(r)}) : x \in Q\}.
\end{align*}

It is straightforward to verify that $\NRD(P^{\sigma} \mid Q^{\sigma}, n) = \NRD(P \mid Q, n)$ as any non-redundant instance $H = (V, E)$ of $P \mid Q$ can be transformed into a non-redundant instance of $P^{\sigma} \mid Q^{\sigma}$ by applying the permutation $\sigma$ to each constraint $e \in E$. Furthermore, since $\sigma$ is invertible, this is a bijection between non-redundant instances.

To prove the lemma, we apply this fact to the various projections of $\BoolBCK \mid \BoolBCKp$. More precisely, for $i \in [9]$, let $J_i := [9] \setminus \{i\}$. For $i \in \{5,9\}$, we construct bijections $\sigma_i : J_i \to J_1$ such that
\[
  (\pi_{J_1}\BoolBCK)^{\sigma_i} \mid (\pi_{J_1}\BoolBCKp)^{\sigma_i} = \pi_{J_i}\BoolBCK \mid \pi_{J_i}\BoolBCKp.
\]
Likewise, for $i \in \{3,4,6,7,8\}$, we construct bijections $\sigma_i : J_i \to J_2$ such that
\[
(\pi_{J_2}\BoolBCK)^{\sigma_i} \mid (\pi_{J_2}\BoolBCKp)^{\sigma_i} = \pi_{J_i}\BoolBCK \mid \pi_{J_i}\BoolBCKp.
\]
We summarize these maps in the following table, with $\bot$ indicating the given input is not in the domain.

\begin{center}
\begin{tabular}{c|ccccccccc}
$j$ & 1 & 2 & 3 & 4 & 5 & 6 & 7 & 8 & 9\\\hline
$\sigma_3(j)$ & 1 & 3 & $\bot$ & 7 & 9 & 8 & 4 & 6 & 5\\
$\sigma_4(j)$ & 1 & 4 & 7 & $\bot$ & 5 & 8 & 3 & 6 & 9\\
$\sigma_5(j)$ & 5 & 2 & 8 & 4 & $\bot$ & 7 & 6 & 3 & 9\\
$\sigma_6(j)$ & 9 & 6 & 3 & 8 & 5 & $\bot$ & 8 & 4 & 1\\
$\sigma_7(j)$ & 1 & 7 & 4 & 3 & 9 & 6 & $\bot$ & 8 & 5\\
$\sigma_8(j)$ & 9 & 3 & 6 & 7 & 1 & 4 & 8 & $\bot$ & 5\\
$\sigma_9(j)$ & 5 & 6 & 4 & 8 & 9 & 7 & 2 & 3 & $\bot$
\end{tabular}
\end{center}
For example, using $\sigma_3$, one can see that
\begin{align*}
(\pi_{J_2} t^{(e)})^{\sigma_3} &= \pi_{J_1} t^{(e)} & (\pi_{J_2} t^{(12)})^{\sigma_3} &= \pi_{J_1} t^{(13)}\\
(\pi_{J_2} t^{(13)})^{\sigma_3} &= \pi_{J_1} t^{(12)} & (\pi_{J_2} t^{(23)})^{\sigma_3} &= \pi_{J_1} t^{(23)}\\
(\pi_{J_2} t^{(123)})^{\sigma_3} &= \pi_{J_1} t^{(321)} & (\pi_{J_2} t^{(321)})^{\sigma_3} &= \pi_{J_1} t^{(123)},
\end{align*}
which certifies that $(\pi_{J_2}\BoolBCK)^{\sigma_3} \mid (\pi_{J_2}\BoolBCKp)^{\sigma_3} = \pi_{J_3}\BoolBCK \mid \pi_{J_3}\BoolBCKp.$ The validity of these remaining bijections can be checked by a similar argument.
\end{proof}

With the symmetry breaking established, it now suffices to prove Lemma~\ref{lem:J1} and Lemma~\ref{lem:J2}.

\begin{lemma}\label{lem:J1}
There exists $\mathcal I = \{I_2, I_3, \hdots, I_9\}$ each of whose elements has cardinality $3$ such that $(C_6^* \mid C_6) \boxtimes (C_6^* \mid C_6)$ is an $\mathcal I$-substructure of $(\pi_{J_1}\BoolBCK \mid \pi_{J_1}\BoolBCKp)$.
\end{lemma}

\begin{proof}
We use the notation $R_2 \mid S_2$ for $(C_6^* \mid C_6) \boxtimes (C_6^* \mid C_6)$. By Proposition~\ref{prop:projection-compact}, it suffices to construct a map $\Sigma : S_2 \to \pi_{J_1} \BoolBCKp$ such that the following conditions hold:
\begin{itemize}
\item $\Sigma(R_2) \subseteq  \pi_{J_1} \BoolBCK$
\item $\Sigma(S_2 \setminus R_2) \subseteq \pi_{J_1} \BoolBCK \setminus \pi_{J_1} \BoolBCK$
\item Every output coordinate of $\Sigma$ depends only on three input coordinates.
\end{itemize}
Using the Satisfiability Modulo Theories (SMT) solver Z3~\cite{Z3}, we construct such a $\Sigma$ as follows, which can be checked by inspection.
\begin{align*}
\Sigma(0, 0, 0, 0) &= (0, 0, 0, 1, 0, 0, 0, 1) & 
\Sigma(0, 1, 0, 0) &= (1, 0, 1, 0, 0, 0, 0, 1)\\
\Sigma(1, 0, 0, 0) &= (0, 1, 0, 1, 0, 1, 0, 0) & 
\Sigma(1, 2, 0, 0) &= (0, 1, 1, 0, 0, 0, 1, 0)\\
\Sigma(2, 1, 0, 0) &= (1, 0, 0, 0, 1, 1, 0, 0) & 
\Sigma(2, 2, 0, 0) &= (0, 0, 0, 0, 1, 0, 1, 0)\\
\Sigma(0, 0, 0, 1) &= (1, 0, 1, 0, 0, 0, 0, 1) & 
\Sigma(0, 1, 0, 1) &= (1, 0, 1, 0, 0, 0, 0, 1)\\
\Sigma(1, 0, 0, 1) &= (1, 0, 0, 0, 1, 1, 0, 0) & 
\Sigma(1, 2, 0, 1) &= (0, 0, 0, 0, 1, 0, 1, 0)\\
\Sigma(2, 1, 0, 1) &= (1, 0, 0, 0, 1, 1, 0, 0) & 
\Sigma(2, 2, 0, 1) &= (0, 0, 0, 0, 1, 0, 1, 0)\\
\Sigma(0, 0, 1, 0) &= (0, 1, 0, 1, 0, 1, 0, 0) & 
\Sigma(0, 1, 1, 0) &= (0, 1, 1, 0, 0, 0, 1, 0)\\
\Sigma(1, 0, 1, 0) &= (0, 1, 0, 1, 0, 1, 0, 0) & 
\Sigma(1, 2, 1, 0) &= (0, 1, 1, 0, 0, 0, 1, 0)\\
\Sigma(2, 1, 1, 0) &= (0, 0, 0, 0, 1, 0, 1, 0) & 
\Sigma(2, 2, 1, 0) &= (0, 0, 0, 0, 1, 0, 1, 0)\\
\Sigma(0, 0, 1, 2) &= (1, 0, 0, 0, 1, 1, 0, 0) & 
\Sigma(0, 1, 1, 2) &= (0, 0, 0, 0, 1, 0, 1, 0)\\
\Sigma(1, 0, 1, 2) &= (1, 0, 0, 0, 1, 1, 0, 0) & 
\Sigma(1, 2, 1, 2) &= (0, 0, 0, 0, 1, 0, 1, 0)\\
\Sigma(2, 1, 1, 2) &= (0, 1, 1, 0, 0, 0, 1, 0) & 
\Sigma(2, 2, 1, 2) &= (0, 1, 1, 0, 0, 0, 1, 0)\\
\Sigma(0, 0, 2, 1) &= (0, 1, 1, 0, 0, 0, 1, 0) & 
\Sigma(0, 1, 2, 1) &= (0, 1, 1, 0, 0, 0, 1, 0)\\
\Sigma(1, 0, 2, 1) &= (0, 0, 0, 0, 1, 0, 1, 0) & 
\Sigma(1, 2, 2, 1) &= (0, 0, 0, 0, 1, 0, 1, 0)\\
\Sigma(2, 1, 2, 1) &= (0, 0, 0, 0, 1, 0, 1, 0) & 
\Sigma(2, 2, 2, 1) &= (0, 0, 0, 0, 1, 0, 1, 0)\\
\Sigma(0, 0, 2, 2) &= (0, 0, 0, 0, 1, 0, 1, 0) & 
\Sigma(0, 1, 2, 2) &= (0, 0, 0, 0, 1, 0, 1, 0)\\
\Sigma(1, 0, 2, 2) &= (0, 0, 0, 0, 1, 0, 1, 0) & 
\Sigma(1, 2, 2, 2) &= (0, 0, 0, 0, 1, 0, 1, 0)\\
\Sigma(2, 1, 2, 2) &= (0, 1, 1, 0, 0, 0, 1, 0) & 
\Sigma(2, 2, 2, 2) &= (0, 1, 1, 0, 0, 0, 1, 0)
\end{align*}
In particular, if we let $(i_2, \hdots, i_9) = (1, 2, 3, 1, 2, 4, 1, 2)$, then the $j$th coordinate of the output does not depend on the $i_j$th coordinate of the input (that is, $I_j = [4] \setminus \{i_j\}$).
\end{proof}

\begin{lemma}\label{lem:J2}
There exists $\mathcal I = \{I_1, I_3, I_4, \hdots, I_9\}$ each of whose elements has cardinality $3$ such that $(C_6^* \mid C_6) \boxtimes (C_6^* \mid C_6)$ is an $\mathcal I$-substructure of $(\pi_{J_2}\BoolBCK \mid \pi_{J_2}\BoolBCKp)$.
\end{lemma}

\begin{proof}
By Proposition~\ref{prop:projection-compact}, it suffices to construct a map $\Sigma : S_2 \to \pi_{J_2} \BoolBCKp$ such that the following conditions hold:
\begin{itemize}
\item $\Sigma(R_2) \subseteq  \pi_{J_2} \BoolBCK$
\item $\Sigma(S_2 \setminus R_2) \subseteq \pi_{J_2} \BoolBCK \setminus \pi_{J_2} \BoolBCK$
\item Every output coordinate of $\Sigma$ depends only on three input coordinates.
\end{itemize}Using the Z3 SMT solver~\cite{Z3}, we construct such a $\Sigma$ as follows, which can be checked by inspection.
\begin{align*}
\Sigma(0, 0, 0, 0) &= (1, 0, 0, 1, 0, 0, 0, 1) & 
\Sigma(0, 1, 0, 0) &= (0, 0, 1, 0, 0, 0, 0, 1)\\
\Sigma(1, 0, 0, 0) &= (1, 0, 0, 0, 1, 0, 1, 0) & 
\Sigma(1, 2, 0, 0) &= (0, 0, 0, 0, 1, 1, 0, 0)\\
\Sigma(2, 1, 0, 0) &= (0, 1, 1, 0, 0, 0, 1, 0) & 
\Sigma(2, 2, 0, 0) &= (0, 1, 0, 1, 0, 1, 0, 0)\\
\Sigma(0, 0, 0, 1) &= (0, 0, 1, 0, 0, 0, 0, 1) & 
\Sigma(0, 1, 0, 1) &= (0, 0, 1, 0, 0, 0, 0, 1)\\
\Sigma(1, 0, 0, 1) &= (0, 1, 1, 0, 0, 0, 1, 0) & 
\Sigma(1, 2, 0, 1) &= (0, 1, 0, 1, 0, 1, 0, 0)\\
\Sigma(2, 1, 0, 1) &= (0, 1, 1, 0, 0, 0, 1, 0) & 
\Sigma(2, 2, 0, 1) &= (0, 1, 0, 1, 0, 1, 0, 0)\\
\Sigma(0, 0, 1, 0) &= (0, 1, 0, 1, 0, 1, 0, 0) & 
\Sigma(0, 1, 1, 0) &= (0, 1, 1, 0, 0, 0, 1, 0)\\
\Sigma(1, 0, 1, 0) &= (0, 0, 0, 0, 1, 1, 0, 0) & 
\Sigma(1, 2, 1, 0) &= (0, 0, 0, 0, 1, 1, 0, 0)\\
\Sigma(2, 1, 1, 0) &= (0, 1, 1, 0, 0, 0, 1, 0) & 
\Sigma(2, 2, 1, 0) &= (0, 1, 0, 1, 0, 1, 0, 0)\\
\Sigma(0, 0, 1, 2) &= (0, 0, 0, 0, 1, 1, 0, 0) & 
\Sigma(0, 1, 1, 2) &= (1, 0, 0, 0, 1, 0, 1, 0)\\
\Sigma(1, 0, 1, 2) &= (0, 0, 0, 0, 1, 1, 0, 0) & 
\Sigma(1, 2, 1, 2) &= (0, 0, 0, 0, 1, 1, 0, 0)\\
\Sigma(2, 1, 1, 2) &= (1, 0, 0, 0, 1, 0, 1, 0) & 
\Sigma(2, 2, 1, 2) &= (0, 0, 0, 0, 1, 1, 0, 0)\\
\Sigma(0, 0, 2, 1) &= (0, 1, 1, 0, 0, 0, 1, 0) & 
\Sigma(0, 1, 2, 1) &= (0, 1, 1, 0, 0, 0, 1, 0)\\
\Sigma(1, 0, 2, 1) &= (0, 1, 1, 0, 0, 0, 1, 0) & 
\Sigma(1, 2, 2, 1) &= (0, 1, 0, 1, 0, 1, 0, 0)\\
\Sigma(2, 1, 2, 1) &= (0, 1, 1, 0, 0, 0, 1, 0) & 
\Sigma(2, 2, 2, 1) &= (0, 1, 0, 1, 0, 1, 0, 0)\\
\Sigma(0, 0, 2, 2) &= (1, 0, 0, 0, 1, 0, 1, 0) & 
\Sigma(0, 1, 2, 2) &= (1, 0, 0, 0, 1, 0, 1, 0)\\
\Sigma(1, 0, 2, 2) &= (1, 0, 0, 0, 1, 0, 1, 0) & 
\Sigma(1, 2, 2, 2) &= (0, 0, 0, 0, 1, 1, 0, 0)\\
\Sigma(2, 1, 2, 2) &= (1, 0, 0, 0, 1, 0, 1, 0) & 
\Sigma(2, 2, 2, 2) &= (0, 0, 0, 0, 1, 1, 0, 0)
\end{align*}
In particular, if we let $(i_1, i_3, i_4, \hdots, i_9) = (1,2,1,3,2,1,4,2)$, then the $j$th coordinate of the output does not depend on the $i_j$th coordinate of the input (that is, $I_j = [4] \setminus \{i_j\}$).
\end{proof}

\section{Conclusion and Open Questions}\label{sec:conclusion}

In this paper, we use a hypergraph projection framework to study the relationship between the non-redundancy of various CSP predicates. By studying what we call \emph{shrinking instances}, we can obtain sharper lower bounds on the non-redundancy of many predicates. The progress we made in this paper serves as the next step toward conquering the biggest open questions concerning the nature of non-redundancy, especially the classification of which predicates have (near-)linear non-redundancy. The use of automated reasoning  highlights the sophistication of the gadget reductions we consider as well the possibility of even more elaborate reductions in the future.

The primary open question we leave is whether super-linear lower bounds exist for $\BoolBCK$. Given how our methods can yield such lower bounds for every nontrivial projection of $\BoolBCK \mid \BoolBCKp$, the primary barrier at this point may be to identify a suitable (conditional) predicate with well-behaved shrinking instances. Given that we found that every nontrivial projection of $\BoolBCK \mid \BoolBCKp$ has $C_6 \boxtimes C_6$ as a substructure, we can infer that the construction of high-girth graphs can yield versatile shrinking instances which can be beneficial in complex gadget reductions. However, even when we tried longer cycles, we were not able to construct a suitable substructure of $\BoolBCK \mid \BoolBCKp$ itself. One potential route to such a substructure would be to consider high-girth \emph{hypergraph} problems rather than graph problems. We leave such investigations as the subject of future work.

\bibliography{catalan-cousins}

\appendix

\section{Finding $\mathcal I$-substructures using a SAT-solver}\label{app:sat}

Given predicates $P_1 \subsetneq Q_1 \subseteq D_1^{r_1}$ and $P_2 \subsetneq Q_2 \subseteq D_2^{r_2}$ as well as a family $\mathcal I = (I_1, \hdots, I_{r_2})$ of subsets of $[r_1]$ a natural question is how one can verify whether $P_1 \mid Q_1$ is a $\mathcal I$-substructure of $P_2 \mid Q_2$. In this appendix, we briefly discuss how this can be encoded into SAT.

From Proposition~\ref{prop:projection-compact}, it suffices to construct a map $\Sigma : Q_1 \to Q_2$ satisfying the following conditions.
\begin{itemize}
\item[(1)] $\Sigma(P_1) \subseteq P_2$ and $\Sigma(Q_1 \setminus P_1) \subseteq Q_2 \setminus P_2$.
\item[(2)] For all $i \in [r_1]$ and $j \in [r_2]$, if $i \not\in I_j$ and $q, q' \in Q_1$ differ only in coordinate $i$, then $\Sigma(q)_j = \Sigma(q')_j$\end{itemize}

For each $q \in Q_1$, $i \in [r_2]$, and $d \in D_2$, we have a variable $x_{q,i,d}$ which represents whether $\Sigma_i(q) = d$. From this, condition (2) can be quickly encoded as equality conditions between some of the variables. More precisely $x_{q,j,d} = x_{q',j,d}$ (aka $(x_{q,j,d} \vee \bar{x}_{q',j,d}) \wedge (\bar{x}_{q,j,d} \vee x_{q',j,d})$) whether $q$ and $q'$ differ only in the $i$th coordinate and $i \not\in I_j$.

We enforce (1) using auxiliary variables. For every $q \in Q_1$ and $q' \in Q_2$, we have a variable $y_{q,q'}$ representing whether $\sigma(q) = q'$. For each such $y_{q,q'}$ we construct a series of implications connecting them to the $x_{q,i,d}$'s.
\[
\forall i \in [r_2], d \in D_2, y_{q,q'} \implies x_{q,i,d} \text{ if } d = q'_i \text{ otherwise } y_{q,q'}\implies \bar{x}_{q,i,d}.
\]
To enforce that $\Sigma(P_1) \subseteq P_2$ and $\Sigma(Q_1 \setminus P_1) \subseteq Q_2 \setminus P_2$, we force $\bar{y_{q,q'}}$ if $q \in P_1$ but $q' \not\in P_1$ or vice-versa. Finally, we need for each $q \in Q_1$ that $y_{q,q'}$ is true for some $q' \in Q_2$ which can be done by a suitable CNF: $\bigvee_{q' \in Q_2} y_{q,q'}.$ This completes a valid SAT-encoding of the $\mathcal I$-substructure condition.

\begin{remark}
This approach is not used directly in this paper. Rather, using the Satisfiability Modulo Theories (SMT) solver Z3~\cite{Z3}, we write out the conditions of being a $\mathcal I$-substructure in a higher level of abstraction (i.e., $\Sigma_1$-sentences involving integer variables) which Z3 then can find a solution for using a variety of techniques including SAT solving.
\end{remark}

\section{Conditional Non-redundancy Can Be Made Non-conditional}\label{app:conditional}

In \cite[Lemma 4.9]{BGJLW25} and \cite[Proposition 2.4]{BGP26}, the authors show for many conditional predicates $P \subsetneq Q \subseteq D^r$, there is a corresponding (non-conditional) predicate $R \subseteq D^{r+1}$ such that $\NRD(R, n) = \Theta(\NRD(P \mid Q, n) \cdot n)$, showing that conditional non-redundancy is not merely a useful abstraction but also a necessary quantity to analyze the non-redundancy of all predicates. In this appendix, we show that their arguments can be adapted to prove a similar phenomenon holds for \emph{all} nontrivial conditional predicates.

Recall from Definition~\ref{def:box-product} the notion of a box product of two predicates.

\begin{theorem}
Let $D$ be a domain of size at least $2$  with $\{0,1\} \subseteq D$. Let $P \subsetneq Q \subseteq D^r$ be a conditional predicate with $P$ non-empty. Further define $R \subsetneq S \subseteq D^{2r}$ such that
\[
  R \mid S = (P \mid Q) \boxtimes (\OR_r \mid \{0,1\}^r)
\]
Then, $\NRD(R, n) = \Theta_r(\NRD(P \mid Q, n) \cdot n^r)$, where $\Theta_r$ hides a factor depending only on $r$.
\end{theorem}

\begin{proof}
First, we prove that $\NRD(R, n) = \Omega_r(\NRD(P \mid Q, n) \cdot n^r)$. To see why, let $H = (V, E)$ be a non-redundant instance of $\CSP(P \mid Q)$ on $n/2$ vertices with $|E| = \NRD(P \mid Q, n/2) = \Omega(P \mid Q, n)$. Let $V'$ be a disjoint set of $n/2$ vertices. Define $H' = (V \cup V', E' := E \times \binom{V'}{r})$ which has $\Omega_r(\NRD(P \mid Q, n) \cdot n^r)$ constraints. We claim that $H'$ is a non-redundant instance of $\CSP(R)$. To see why, fix an edge $e' \in E'$, we need to find an assignment $\psi: V \cup V' \to D$ which satisfies every constraint of $H'$ except $e'$. Observe that $\pi_{[r]} e'$ is an edge of $H$, thus by the non-redundancy of $H$ as instance of $\CSP(P \mid Q)$ there is an assignment $\phi : V \to D$ which satisfies every edge of $H$ except $\pi_{[r]} e'$. We can thus define $\psi$ as follows
\[
  \psi(v) :=\begin{cases}
    \phi(v) & v \in V\\
    0 & v = e'_i\text{ for some }i \in \{r+1, \hdots, 2r\}\\
    1 & \text{otherwise}
    \end{cases}
\]
To see why this works, recall that $R = (P \times \{0,1\}^r) \cup (Q \times \OR_r)$. Thus, for any $e'' \in E' \setminus \{e'\}$ with $\pi_{[r]} {e''} \neq \pi_{[r]} e'$, we have that $\psi(e'')\in P \times \{0,1\}^r$.  Likewise, if $e'' \in E' \setminus \{e'\}$ with $\pi_{\{r+1, \hdots, 2r\}} {e''} \neq \pi_{\{r+1, \hdots, 2r\}} e$, then $\psi(e'') \in Q \times \OR_r$. However, $\psi(e') = (\phi(\pi_{[r]}e'), 0, \hdots, 0) \not\in R$. This shows that $H'$ is a non-redundant instance of $\CSP(R)$.

We now turn to showing that $\NRD(R, n) = O_r(\NRD(P \mid Q, n) \cdot n^r)$ By Lemma~\ref{lem:NRD-triangle}, we have that
\[
  \NRD(R, n) \le \NRD(R \mid S, n) + \NRD(S, n).
\]
Thus, it suffices to show that $\NRD(R \mid S, n) =  O_r(\NRD(P \mid Q, n) \cdot n^r)$ and $\NRD(S, n) =  O_r(\NRD(P \mid Q, n) \cdot n^r)$.

For the former upper bound, let $H = (V, E \subseteq V^{2r})$ be a non-redundant instance of $\CSP(R \mid S)$ with $|V| = n$ and $|E| = \NRD(R \mid S, n)$. By the pigeonhole principle, there must be some string $s \in V^r$ such that $\pi_{\{r+1, \hdots, 2r\}} e = s$ for at least $|E|/n^r$ choices of $e \in E$. Define $H' = (V, E' \subseteq V^r)$ where
\[
  E' := \{\pi_{[r]} e : e \in E, \pi_{\{r+1, \hdots, 2r\}} e = s\}.
\]
We claim that $H'$ is a non-redundant instance of $\CSP(P \mid Q)$. Pick distinct $e'_1, e'_2 \in E'$ and let $e_1, e_2 \in E$ be the corresponding edge in $H$. Since $H$ is non-redundant, there is an assignment $\psi : V \to D$ which satisfies every edge of $H$ except $e_1$. In particular, $\psi(e_2) \in R = (P \times \{0,1\}^r) \cup (Q \times \OR_r)$ but $\psi(e_1) \in S \setminus R = (Q \setminus P) \times \{0^r\}$. Since $e_1$ and $e_2$ share the same last $r$ coordinates, we must have that $\psi(e_2) \in P \times \{0^r\}$. In other words (returning to $H'$), $\psi(e'_1) \in Q \setminus P$ but $\psi(e'_2) \in P$ for every $e'_2 \in E' \setminus \{e'_1\}$. This proves that $H'$ is a non-redundant instance of $\CSP(P \mid Q)$. Therefore, $\NRD(P \mid Q, n) = \Omega(\NRD(R \mid S, n) / n^r)$, as desired.

As a last step, we must bound $\NRD(S, n)$. Recall that $S = Q \times \{0,1\}^r$. Consider a non-redundant instance $H = (V, E \subseteq V^{2r})$ of $\CSP(S)$. Given a string $s \in V^r$, let $H_s = (V, E_s)$ be the subinstance where
\[
  E_s := \{e \in E : \pi_{[r]} e = s\}.
\]
Then, we claim that $|E_s| \le n$. To see why, fix any $e \in E_s$ and let $\psi : V \to D$ be a map which fails to satisfy $e$ but satisfies all other $e' \in E_s$. Since $\pi_{[r]} e' = \pi_{[r]} e = s$, we must have that $\psi(s) \in Q$. Therefore, for $\psi(e) \not \in S$, we must have that $\psi(\pi_{\{r+1, \hdots, 2r\}} e) \not\in \{0,1\}^r$. That is, $\psi(e_i) \not\in \{0,1\}$ for some $i \in \{r+1, \hdots, 2r\}$. Crucially, this $e_i$ cannot appear as a vertex in $\pi_{\{r+1, \hdots, 2r\}} e'$ for any $e' \in E_s \setminus \{e\}$ (or else $\psi(e') \in S$ would be false). Therefore, there exists an injection from $E_s$ to $V$ by picking this vertex which maps outside of $\{0,1\}$. This proves that $|E_s| \le n$.

Since there are $n^r$ choices for $s \in V^r$, we have proved that $\NRD(S, n) \le n^{r+1}$. Recall that $P$ is non-empty. Thus, $\NRD(P \mid Q, n) \ge n/r = \Omega_r(n)$ by considering an instance with $n/r$ disjoint clauses. Thus, 
\[
  \NRD(S, n) \le n^{r+1} = O_r(\NRD(P \mid Q, n) \cdot n^r).
\]
This fact completes the proof.
\end{proof}

\section{Lower Bounds for Non-Boolean Predicates} \label{app:non-Boolean}

Going beyond Theorem~\ref{thm:BoolBCK-projections}, we also use Theorem~\ref{thm:shrink-lb} to deduce novel super-linear lower bounds for \emph{non-Boolean} predicates. Very recently, \cite{BGJLW25} developed an analogue of the Boolean notion of balanced which applies to any domain by a notion called the \emph{Catalan identities}. To illustrate this notion, consider a finite domain $D$ and a word $w \in D^{\ell}$, where $\ell$ is an odd integer. We can play the following ``cancellation game'' on $w$. If $w$ has two consecutive symbols which are equal, delete both, and recurse until no more progress can be made. For example, if $D = \{0,1,2\}$ then playing the game with $0221221$ yields
\[
    0221221 \to 01221 \to 011 \to 0.
\]
It turns out the outcome of this game is always the same no matter the order the cancellations are done. Furthermore, this cancellation game can be used to determine whether a predicate has a Mal'tsev embedding. As an example, consider the predicate $\Cat_5 \subseteq \{0,1,2\}^5$ defined by
\[
  \Cat_5 = \{01012, 11111, 12201, 22222, 20120\}.
\]
Now build the following matrix whose \emph{columns} are tuples in $\Cat_5$.
\begin{align}
\begin{pmatrix}
0 & 1 & 1 & 2 & 2\\
1 & 1 & 2 & 2 & 0\\
0 & 1 & 2 & 2 & 1\\
1 & 1 & 0 & 2 & 2\\
2 & 1 & 1 & 2 & 0
\end{pmatrix} \label{eq:cat5}
\end{align}
If one plays the cancellation game on each \emph{row} of this matrix, one is left with the tuple $00000$ which is \emph{not} in $\Cat_5$. By a deep result of \cite{BGJLW25}, this implies that $\Cat_5$ lacks a Mal'tsev embedding over any domain, so we have no evidence that $\Cat_5$ has (near-)linear NRD. In fact, similar to how $\BoolBCK$ is the `canonical' non-balanced predicate with five tuples, $\Cat_5$ is the canonical predicate with five tuples which lacks a Mal'tsev embedding. More precisely, the rows of (\ref{eq:cat5}) correspond to every possible cancellation game on $5$-tuples. As such, one can show that any non-redundancy lower bound of $\CSP(\Cat_5)$ implies a non-redundancy lower bound for $\CSP(\BoolBCK)$.This is witnessed by the map $\Sigma : \Cat_5^{+} \to \BoolBCKp$ defined by
\begin{align*}
  \Sigma(00000) &= 100010001\\
  \Sigma(01012) &= 010100001\\
  \Sigma(11111) &= 001100010\\
  \Sigma(12201) &= 001010100\\
  \Sigma(22222) &= 010001100\\
  \Sigma(20120) &= 100001010
\end{align*}

Therefore, resolving the non-redundancy of $\Cat_5$ is an essential question toward determining which predicates have linear non-redundancy.

Similar to $\BoolBCK$, if we let $\Cat_5^{+} = \Cat_5 \cup \{00000\}$, then $\Cat_5^{+}$ can be shown to have linear non-redundancy. Thus, we can always assume every non-assignment to a non-redundant instance of $\CSP(\Cat_5^{+})$ is $00000$, reducing our problem to the non-redundancy of the conditional predicate $\Cat_5 \mid \Cat_{5}^{+}$.  Although proving super-linear lower bounds for $\Cat_5$ is harder than for $\BoolBCK$, we still give a super-linear lower bound for one of the maximal non-trivial projections of $\Cat_5 \mid \Cat_{5}^{+}$.

\begin{theorem}[See Theorem~\ref{thm:catalan-cousins}]\label{thm:cat5}
$\NRD(\pi_{[4]}\Cat_5 \mid \pi_{[4]}\Cat_5^{+}, n) \ge n^{6/5}$.
\end{theorem}

This result can be viewed as one step closer toward showing that Mal'tsev embeddings are equivalent to having (near-)linear non-redundancy.

\subsection{Lower Bounds}

In addition to $\pi_{[4]}\Cat_5 \mid \pi_{[4]}\Cat_5^{+}$, we also give stronger lower bounds for some arity $3$ projections. More precisely, consider the following predicates. For ease of discussion, we give each succinct names.

\begin{align*}
P_1 &:= \pi_{\{1,3,4\}} \Cat_5 = \{001,111,120,222,212\}\\
Q_1 &:= \pi_{\{1,3,4\}} \Cat_5^{+} = P_1 \cup \{000\}\\
P_2 &:= \pi_{\{1,2,4\}} \Cat_5 = \{011,111,120,222,202\}\\
Q_2 &:= \pi_{\{1,2,4\}} \Cat_5^{+} = P_2 \cup \{000\}\\
P_3 &:= \pi_{\{1,2,3,4\}} \Cat_5 = \{0101, 1111, 1220, 2222, 2012\}\\
Q_3 &:= \pi_{\{1,2,3,4\}} \Cat_5^{+} = P_3 \cup \{0000\}
\end{align*}

\begin{remark}
A projection  we do \emph{not} consider in this paper is
\[
    \pi_{\{3,4,5\}} \Cat_5 = \{012,111,201,222,120\},
\]
which is equivalent to a ternary predicate identified by~\cite{BCK20} which has lower non-redundancy than $\BoolBCK$. See Section 7 of \cite{BG25} for further discussion.
\end{remark}

The goal of this section is to prove the following lower bounds.
\begin{theorem}\label{thm:catalan-cousins}
We following lower bounds hold.
\begin{align}
\NRD(P_1 \mid Q_1, n) &\ge \Omega(n^{4/3}) \label{eq:P1Q1}\\
\NRD(P_2 \mid Q_2, n) &\ge \Omega(n^{4/3}) \label{eq:P2Q2}\\
\NRD(P_3 \mid Q_3, n) &\ge \Omega(n^{6/5}) \label{eq:P3Q3}
\end{align}
\end{theorem}
\begin{proof}
By applying Lemma~\ref{lem:R1S1} and Lemma~\ref{lem:R2S2} to Theorem~\ref{thm:shrink-lb} with $\eps = 1/4$ and $\eps = 1/6$, respectively, it suffices to prove Lemma~\ref{lem:P1Q1}, Lemma~\ref{lem:P2Q2}, and Lemma~\ref{lem:P3Q3} to establish (\ref{eq:P1Q1}), (\ref{eq:P2Q2}), and (\ref{eq:P3Q3}), respectively.
\end{proof}

\begin{lemma}\label{lem:P1Q1}
There exists $\mathcal I = \{I_1, I_2, I_3\}$ each of size $3$ such that $R_1 \mid S_1$ is an $\mathcal I$-substructure of $P_1 \mid Q_1$.
\end{lemma}
\begin{proof}
By Proposition~\ref{prop:projection-compact}, it suffices to construct a map $\Sigma : S_1 \to Q_1$ such that the following conditions hold:
\begin{itemize}
\item $\Sigma(R_1) \subseteq  P_1$
\item $\Sigma(S_1 \setminus R_1) \subseteq Q_1 \setminus P_1$
\item Every output coordinate of $\Sigma$ depends only on two input coordinates.
\end{itemize}
Using the Z3 SMT solver~\cite{Z3}, we constructed such a $\Sigma$ as follows, which can be checked by inspection.
\begin{align*}
\Sigma(0, 0, 0) &= (0, 0, 0) & 
\Sigma(0, 0, 1) &= (0, 0, 1)\\
\Sigma(0, 0, 2) &= (0, 0, 1) & 
\Sigma(0, 1, 0) &= (1, 2, 0)\\
\Sigma(0, 1, 1) &= (1, 1, 1) & 
\Sigma(0, 1, 2) &= (1, 1, 1)\\
\Sigma(1, 0, 0) &= (0, 0, 1) & 
\Sigma(1, 0, 1) &= (0, 0, 1)\\
\Sigma(1, 0, 2) &= (0, 0, 1) & 
\Sigma(1, 2, 0) &= (1, 1, 1)\\
\Sigma(1, 2, 1) &= (1, 1, 1) & 
\Sigma(1, 2, 2) &= (1, 1, 1)\\
\Sigma(2, 1, 0) &= (2, 2, 2) & 
\Sigma(2, 1, 1) &= (2, 1, 2)\\
\Sigma(2, 1, 2) &= (2, 1, 2) & 
\Sigma(2, 2, 0) &= (2, 1, 2)\\
\Sigma(2, 2, 1) &= (2, 1, 2) & 
\Sigma(2, 2, 2) &= (2, 1, 2)
\end{align*}
Here $(I_1, I_2, I_3) = (\{1,2\},\{2,3\},\{1,3\}).$
\end{proof}

\begin{lemma}\label{lem:P2Q2}
There exists $\mathcal I = \{I_1, I_2, I_3\}$ each of size $3$ such that $R_1 \mid S_1$ is an $\mathcal I$-substructure of $P_2 \mid Q_2$.
\end{lemma}
\begin{proof}
By Proposition~\ref{prop:projection-compact}, it suffices to construct a map $\Sigma : S_1 \to Q_2$ such that the following conditions hold:
\begin{itemize}
\item $\Sigma(R_1) \subseteq  P_2$
\item $\Sigma(S_1 \setminus R_1) \subseteq Q_2 \setminus P_2$
\item Every output coordinate of $\Sigma$ depends only on two input coordinates.
\end{itemize}
Using the Z3 SMT solver~\cite{Z3}, we constructed such a $\Sigma$ as follows, which can be checked by inspection.
\begin{align*}
\Sigma(0, 0, 0) &= (0, 0, 0) & 
\Sigma(0, 0, 1) &= (1, 0, 2)\\
\Sigma(0, 0, 2) &= (1, 0, 2) & 
\Sigma(0, 1, 0) &= (2, 2, 0)\\
\Sigma(0, 1, 1) &= (2, 2, 2) & 
\Sigma(0, 1, 2) &= (2, 2, 2)\\
\Sigma(1, 0, 0) &= (0, 1, 1) & 
\Sigma(1, 0, 1) &= (1, 1, 1)\\
\Sigma(1, 0, 2) &= (1, 1, 1) & 
\Sigma(1, 2, 0) &= (1, 1, 1)\\
\Sigma(1, 2, 1) &= (1, 1, 1) & 
\Sigma(1, 2, 2) &= (1, 1, 1)\\
\Sigma(2, 1, 0) &= (2, 2, 2) & 
\Sigma(2, 1, 1) &= (2, 2, 2)\\
\Sigma(2, 1, 2) &= (2, 2, 2) & 
\Sigma(2, 2, 0) &= (1, 0, 2)\\
\Sigma(2, 2, 1) &= (1, 0, 2) & 
\Sigma(2, 2, 2) &= (1, 0, 2)
\end{align*}
Here $(I_1, I_2, I_3) = (\{1,2\},\{2,3\},\{1,3\}).$
\end{proof}

\begin{remark}\label{remark:better}
For $\CSP(P_1 \mid Q_1)$ and $\CSP(P_2 \mid Q_2)$, since the projections sets have size at most $2$ and $\NRD(R_1 \mid S_1, n) = \Omega(n^{5/2})$, the methods of \cite{BGJLW25} can be used to immediately prove a lower bound of $\Omega(n^{5/4})$. However, using our shrinking instances, we get a sharper lower bound of $\Omega(n^{4/3})$.
\end{remark}

\begin{lemma}\label{lem:P3Q3}
There exists $\mathcal I = \{I_1, I_2, I_3, I_4\}$ each of size $3$ such that $R_2 \mid S_2$ is an $\mathcal I$-substructure of $P_3 \mid Q_3$.
\end{lemma}

\begin{proof}
By Proposition~\ref{prop:projection-compact}, it suffices to construct a map $\Sigma : S_2 \to Q_3$ such that the following conditions hold:
\begin{itemize}
\item $\Sigma(R_2) \subseteq  P_3$
\item $\Sigma(S_2 \setminus R_2) \subseteq Q_3 \setminus P_3$
\item Every output coordinate of $\Sigma$ depends only on three input coordinates.
\end{itemize}
Using the Z3 SMT solver~\cite{Z3}, we constructed such a $\Sigma$ as follows, which can be checked by inspection.
\begin{align*}
\Sigma(0, 0, 0, 0) &= (0, 0, 0, 0) & 
\Sigma(0, 0, 0, 1) &= (1, 2, 2, 0)\\
\Sigma(0, 0, 1, 0) &= (0, 1, 0, 1) & 
\Sigma(0, 0, 1, 2) &= (1, 1, 1, 1)\\
\Sigma(0, 0, 2, 1) &= (2, 2, 2, 2) & 
\Sigma(0, 0, 2, 2) &= (2, 0, 1, 2)\\
\Sigma(0, 1, 0, 0) &= (2, 0, 1, 2) & 
\Sigma(0, 1, 0, 1) &= (2, 2, 2, 2)\\
\Sigma(0, 1, 1, 0) &= (1, 1, 1, 1) & 
\Sigma(0, 1, 1, 2) &= (1, 1, 1, 1)\\
\Sigma(0, 1, 2, 1) &= (2, 2, 2, 2) & 
\Sigma(0, 1, 2, 2) &= (2, 0, 1, 2)\\
\Sigma(1, 0, 0, 0) &= (0, 1, 0, 1) & 
\Sigma(1, 0, 0, 1) &= (1, 1, 1, 1)\\
\Sigma(1, 0, 1, 0) &= (0, 1, 0, 1) & 
\Sigma(1, 0, 1, 2) &= (1, 1, 1, 1)\\
\Sigma(1, 0, 2, 1) &= (2, 0, 1, 2) & 
\Sigma(1, 0, 2, 2) &= (2, 0, 1, 2)\\
\Sigma(1, 2, 0, 0) &= (1, 1, 1, 1) & 
\Sigma(1, 2, 0, 1) &= (1, 1, 1, 1)\\
\Sigma(1, 2, 1, 0) &= (1, 1, 1, 1) & 
\Sigma(1, 2, 1, 2) &= (1, 1, 1, 1)\\
\Sigma(1, 2, 2, 1) &= (2, 0, 1, 2) & 
\Sigma(1, 2, 2, 2) &= (2, 0, 1, 2)\\
\Sigma(2, 1, 0, 0) &= (2, 2, 2, 2) & 
\Sigma(2, 1, 0, 1) &= (2, 2, 2, 2)\\
\Sigma(2, 1, 1, 0) &= (1, 2, 2, 0) & 
\Sigma(2, 1, 1, 2) &= (1, 2, 2, 0)\\
\Sigma(2, 1, 2, 1) &= (2, 2, 2, 2) & 
\Sigma(2, 1, 2, 2) &= (2, 2, 2, 2)\\
\Sigma(2, 2, 0, 0) &= (1, 2, 2, 0) & 
\Sigma(2, 2, 0, 1) &= (1, 2, 2, 0)\\
\Sigma(2, 2, 1, 0) &= (1, 2, 2, 0) & 
\Sigma(2, 2, 1, 2) &= (1, 2, 2, 0)\\
\Sigma(2, 2, 2, 1) &= (2, 2, 2, 2) & 
\Sigma(2, 2, 2, 2) &= (2, 2, 2, 2)
\end{align*}
Here $(I_1, I_2, I_3, I_4) = (\{2,3,4\},\{1,3,4\},\{1,2,4\},\{1,2,3\})$.
\end{proof}

\end{document}